\documentclass[%
reprint,
superscriptaddress,
showkeys,
amsmath,amssymb,
aps,
]{revtex4-2}

\usepackage{microtype}
\usepackage{array}
\usepackage{multirow}
\usepackage{diagbox}
\usepackage{amsthm}
\usepackage{xcolor}
\usepackage{verbatim}
\usepackage{url}
\usepackage{graphicx}
\usepackage{dcolumn}
\usepackage{bm}
\usepackage[caption=false]{subfig}
\usepackage{float}

\usepackage{tikz}
\usepackage{quantikz}   
\usepackage{makecell}
\allowdisplaybreaks

\usepackage{algorithm}
\usepackage[noend]{algorithmic}
\newtheorem{theorem}{Theorem}[section] 

\theoremstyle{definition} 
\newtheorem{definition}{Definition}[section]

\theoremstyle{definition} 
\newtheorem{example}{Example}[section] 
\usepackage{hyperref}
\hypersetup{
	colorlinks=true,
	linkcolor=blue,
	citecolor=blue,
}
\begin{document}
	
	\preprint{APS/123-QED}
\title{Toward Quantum Advantage in Learning Parities with Structured Noise \\via Lower Bound Optimization of the Condition Number}
\thanks{A footnote to the article title}%

\author{Yusen Han}
\email{yusenhan01@163.com}

\author{Xuelian Li}
\email{xlli@mail.xidian.edu.cn}
\affiliation{School of Mathematics and Statistics, Xidian University, Xi'an, China}

\author{Juntao Gao}
\email{jtgao@mail.xidian.edu.cn}
\affiliation{School of Telecommunications and Engineering, Xidian University, Xi'an, China}

\author{Bo Song}
\email{songbo1@chinatelecom.cn}
\affiliation{China Telecom Quantum Information Technology Group Co., Ltd, Hefei, China}

\begin{abstract}
	Learning Parities with Structured Noise (LPSN) can be reduced to solving nonlinear Boolean systems. In quantum computing, such systems are typically transformed into Macaulay linear systems and solved via quantum linear system algorithms, a process severely limited by the condition number. To address this, we propose a novel reduction method for Macaulay linear systems. Under the assumptions of Ding et al., we derive a condition number lower bound incorporating a scaling factor. 
	This reduction not only guarantees efficient quantum state preparation but also exhibits a distinct advantage regarding the condition number interval relative to the reduced right-hand side vector, thereby reducing the lower bound of the condition number and ultimately optimizing the upper bound on the time complexity of the quantum algorithm for solving Boolean systems.
	Furthermore, applying this improved quantum algorithm to LPSN significantly reduces sample complexity by exploiting the Macaulay system's solution structure. We further provide a concrete logical-level quantum resource estimate, demonstrating that the optimized condition number translates directly into a reduction in circuit width, depth, and gate count. Finally, we establish an algorithm selection strategy by systematically comparing quantum and classical approaches across noise pattern adaptability, sample complexity, and time complexity. Results demonstrate that our quantum algorithm exhibits the potential to outperform classical counterparts under specific parameter regimes.
	\par 
\end{abstract}

\maketitle


\section{Introduction}
\label{section1}
The \textbf{Learning Parities with Noise (LPN)} problem is a core hardness assumption in computational learning and cryptography. Standard LPN requires recovering a secret vector \(\bm{s}\) from samples \((\bm{a}, \bm{a}\cdot \bm{s}\oplus\eta)\). Despite providing post-quantum security for numerous cryptographic systems \cite{alekhnovich2003more, blum1993cryptographic, hopper2001secure, couteau2025multiparty, li2025efficient}, the cumbersome error-correction necessitated by purely random noise severely limits practical efficiency. 
\par
This paper focuses on the \textbf{Learning Parities with Structured Noise (LPSN)} problem introduced by Arora and Ge \cite{arora2011new}. Its core innovation lies in constraining the noise vector \(\bm{\eta}\) with a polynomial \(P\) such that \(P(\bm{\eta})=0\). Given oracle samples \((A^T, A^T\bm{s} \oplus \bm{\eta})\), this polynomial structure significantly improves mathematical tractability while preserving security, offering a novel pathway for efficient cryptographic protocol design.

\par
\noindent \textbf{Previous Work.} Existing algorithms for the LPN problem face efficiency bottlenecks. In the high-noise regime, Blum et al.~\cite{blum2003noise} proposed the state-of-the-art algorithm with \(\mathcal{O}(2^{\Theta(\frac{n}{\log_2 n})})\) complexity. In the low-noise regime, mainstream Gaussian elimination methods require approximately \(\mathcal{O}(e^{np})\) time. Consequently, practical cryptographic designs~\cite{wang2012advances, langlois2014lattice, yu2019collision} generally assume LPN remains secure when the noise rate is below \(1/\sqrt{n}\).
\par
Quantum computing provides a new perspective. Inspired by Blum et al.~\cite{blum2003noise}, Esser et al.~\cite{cryptoeprint:1990/001} proposed a hybrid quantum attack, subsequently optimized by Jiao~\cite{jiao2020specifications}. In Information Set Decoding approaches, Bernstein~\cite{bernstein2010grover} improved the complexity of the classical Prange~\cite{prange1962use} and May-Ozerov~\cite{may2015computing} algorithms; the current optimal exponent derives from the quantum random walk algorithm by Kachigar and Tillich~\cite{kachigar2017quantum}. Additionally, Tran and Vaudenay~\cite{tran2022solving} combined Gaussian elimination with Grover's search, achieving non-asymptotic advantages in large-dimension, low-noise scenarios.
\par
Regarding the LPSN problem, Arora and Ge~\cite{arora2011new}, building on linearization techniques from~\cite{bard2009algebraic}, proposed classical polynomial-time recovery algorithms, provided that sufficient samples are available and the number of correlated noises is constant. Because the LPSN problem does not necessitate a vector sparsity assumption, its solving algorithms exhibit broad applicability. However, classical efficiency has reached its limits. Exploring the hardness of LPSN under small sample sizes and its potential for quantum acceleration has thus become a critical open problem.
\par 
The algebraic approach to the LPSN problem requires processing large-scale Boolean systems, for which quantum linear system algorithms (QLSA) offer a breakthrough. Following the foundational HHL algorithm~\cite{harrow2009quantum} featuring logarithmic complexity, this field has advanced rapidly: Ambainis~\cite{ambainis2010variable} introduced variable-time amplitude amplification to reduce condition number dependence to linear; Berry et al.~\cite{berry2015hamiltonian} utilized Hamiltonian simulation via truncated Taylor series to lower sparsity dependence to linear; Childs et al.~\cite{childs2017quantum} achieved an exponential improvement in precision dependence. For dense matrices, Wossnig et al.~\cite{wossnig2018quantum} bypassed sparsity constraints using quantum random access memory, while Chakraborty et al.~\cite{chakraborty2018power} exponentially enhanced precision by integrating block-encoding with variable-time amplitude estimation. Furthermore, Lin et al.~\cite{lin2020optimal} achieved near-optimal query complexity for sparse matrices, and Nghiem~\cite{nghiem2025new} recently eliminated the direct dependence on the condition number.
\par
\par
Beyond asymptotic complexity, the practical feasibility of QLSA-based algorithms is ultimately dictated by their concrete physical resource requirements. Scherer et al.~\cite{scherer2017concrete} presented the first concrete logical-level resource estimate for the HHL algorithm, taking the computation of the electromagnetic scattering cross section of a 2D target as a benchmark. Their analysis shows that, even in the idealized logical-circuit model, the algorithm demands several hundred logical qubits together with a gate count and circuit depth exceeding the order of \(10^{25}\); crucially, these resources are dominated by the Hamiltonian simulation and depend critically on the condition number \(\kappa\) and the target precision \(\varepsilon\), while scaling only logarithmically with the matrix size. This indicates that the condition number, rather than the problem size itself, fundamentally governs the resource cost of QLSA-based algorithms.
 
\par
Chen and Gao~\cite{chen2022quantum} pioneered a quantum algorithm for solving nonlinear Boolean systems by transforming them into complex Macaulay linear systems, where the solving efficiency heavily relies on the matrix condition number, the magnitude of which remains difficult to evaluate. To quantitatively evaluate the condition number, Ding et al.~\cite{ding2023limitations} proposed a reduction method, proving that for a system comprising \(n\) Boolean variables, assuming its solution has a Hamming weight of \(h\), the condition number lower bound for a reduced Macaulay matrix of degree \(D\) is \(\Omega(({\binom{D+h}{h}-1})^{\frac{1}{2}})\). When \(h=\Theta(n)\), the exponential growth of the condition number renders this algorithm inferior to Grover's search. To address this, they introduced the Boolean Macaulay linear system to reduce the lower bound to \(\Omega((2^h-1)^{\frac{1}{2}})\). This improvement ensures that the condition number remains polynomial when \(h=\mathcal{O}(\log_2 n)\), thereby retaining the potential for super-polynomial quantum speedup.
\par 

\noindent \textbf{Our Contributions.} This paper makes multidimensional contributions to solving the LPSN problem. By integrating the binary Schwartz-Zippel lemma with inequality bounding techniques, we derive the exact sample and time complexities required for the classical bit-by-bit guessing and algebraic Gaussian elimination algorithms to achieve a success probability of \(1-\varepsilon\), thereby explicitly delineating their theoretical boundaries.
\par
To reduce the lower bound of the condition number in Chen and Gao's quantum algorithm~\cite{chen2022quantum}, we propose a novel polynomial system reduction method. For a nonlinear Boolean system \(\mathcal{F}\) comprising \(r\) equations and \(n\) variables, featuring a sparsity of \(T_{\mathcal{F}}\) and a unique solution of Hamming weight \(h\), this method transforms it into an equivalent 3-sparse complex polynomial system where all constant terms are strictly \(-1\). This operation ensures the Macaulay matrix strictly satisfies the prerequisites for efficient quantum state preparation within \(\mathcal{O}(T_{\mathcal{F}}-2r)\) time. Bridging the theoretical gap left by Ding et al.~\cite{ding2023limitations} regarding unquantified reduction impacts, we optimize the condition number lower bound for a Macaulay matrix of total degree \(D\) from \(\Omega(({\binom{D+h}{h}-1})^{\frac{1}{2}})\) to \(\Omega((T_{\mathcal{F}}-2r)^{-\frac{1}{2}}({\binom{D+h}{h}-1})^{\frac{1}{2}})\).
\par
This reduction framework naturally extends to Boolean Macaulay linear systems, further decreasing the condition number lower bound to \(\Omega (T_{\mathcal{F}}^{-\frac{1}{2}}(2^h-1)^{\frac{1}{2}} )\). Furthermore, we empirically validate these tighter bounds by constructing numerical matrix instances that breach the theoretical limits established by Ding et al. using an iterative optimization algorithm.
\par
Building on the resource-estimation perspective of Scherer et al.~\cite{scherer2017concrete}, we further carry out a fine-grained logical-level resource analysis of the proposed Macaulay and Boolean-Macaulay quantum algorithms. We derive explicit closed-form expressions for the circuit width, depth, and gate count, and show that, since the circuit size scales as the square of the condition number, the condition-number optimization achieved by our reduction translates directly into a reduction in circuit depth and gate count. This establishes that our improvement is not merely asymptotic but yields concrete, quantifiable resource savings.
\par
Ultimately, applying the improved quantum algorithm to the LPSN problem achieves an exponential reduction in sample complexity. By comparing time complexities, we provide a comprehensive algorithm selection strategy. Assisted by 3D time complexity surface plots, we reveal that even when \(h=\Theta(\sqrt{n})\), the quantum algorithm possesses the potential to outperform classical algorithms, this further highlighting its theoretical supremacy within specific parameter regimes.
\par 
 
\noindent \textbf{Structure of the Paper.} The remainder of this paper is organized as follows. Section \ref{section2} introduces the condition number of a matrix, the linearization technique utilized in classical algebraic attacks, as well as the complex-field transformation of Boolean equation systems and the construction of linear systems within the framework of quantum algebraic attacks. Section \ref{section3} formalizes the LPN and LPSN problems based on the oracle model. Section \ref{section4} provides a quantitative analysis of classical algorithms, deriving their sample and time complexities under a specified success probability. Section \ref{section5} details the improved quantum algorithm for solving Boolean systems and evaluates its resource requirements when applied to the LPSN problem. Section \ref{section6} provides a concrete evaluation of the quantum resources required by the proposed algorithms. Section \ref{section7} experimentally validates the correctness of the theoretical analysis and proposes algorithm selection strategies tailored to different parameter regimes based on the comparative results. Finally, Section \ref{section8} summarizes the content of this paper.

\section{Preliminaries}
\label{section2}
\subsection{Condition Number}
\label{section2-1}
The standard condition number of a matrix $A$, defined as $\kappa(A)=\|A\|_2\cdot\|A^+\|_2$, reflects the global worst-case sensitivity of a linear system. Hereafter, we uniformly utilize $\|A\|$ to denote the spectral norm. The condition number with respect to a specific vector $\bm{b}$ is defined as $\kappa_{\bm{b}}(A)=\|A\|\cdot \frac{\|A^+\bm{b}\|}{\|\bm{b}\|}$. Because $\bm{b}$ does not necessarily align with the left singular vector corresponding to the minimum singular value, the inequality $\kappa_{\bm{b}}(A)\leq\kappa(A)$ generally holds. This metric provides a tighter error bound and correlates directly with the practical runtime of quantum algorithms.
\par
The truncated condition number $\kappa_k(A)$ exclusively measures the ratio involving the top $k$ largest singular values. Truncated quantum algorithms attempt to bypass ill-conditioned spectra by operating exclusively within a well-conditioned subspace $S$ defined by singular values bounded below by $\frac{c}{\kappa_{\bm{b}}(A)}$ for $c>1$. However, the projection of the solution vector $\bm{x}$ onto $S$ strictly satisfies $\|\Pi_S\bm{x}\|\leq\frac{\kappa_{\bm{b}}(A)\|\bm{b}\|}{c}$. This implies that the principal components of the solution predominantly reside outside this well-conditioned subspace, or the system remains acutely sensitive to perturbations due to a diminished overall norm.
\par
Consequently, to accurately capture the solution, truncated algorithms must lower their truncation thresholds, fundamentally failing to evade the complexity bottleneck dictated by $\kappa_{\bm{b}}(A)$. Because the global bound $\kappa(A)$ cannot accurately assess the runtime of truncated variants, our subsequent analysis focuses exclusively on $\kappa_{\bm{b}}(A)$. This metric establishes a rigorous theoretical foundation for evaluating the ultimate performance limits of all QLSA and their truncated variants.

\subsection{Linearization Techniques in Classical \\Algebraic Attacks}
\label{section2-2}
In the randomized LPSN setting, an attacker treats $\bm{x}\in \mathrm{GF}(2)^n$ as a variable vector and queries the oracle $Q_1$ to obtain $m$ samples, denoted as $(\bm{a}^T_1,b_1)^T, \dots, (\bm{a}^T_m,b_m)^T$. Given the noise correlation polynomial $P(\bm{\eta})=0$, substituting $\eta_i=\bm{a}_i\cdot \bm{x}\oplus b_i$ yields a degree-$d$ polynomial constraint over $\bm{x}$. Arora and Ge applied straightforward linearization by replacing each monomial $\prod_{i\in S}x_i$, where $S\subseteq[n]$ and $|S|\leq d$, with a new variable $y_S$. This converts the polynomial into a linear constraint over a new vector $\bm{y}$ comprising $N=\sum_{i=1}^d\binom{n}{i}$ variables, where $y_{\emptyset}=1$ is treated as a constant:
$
L(P(\bm{a}_1\cdot \bm{x}\oplus b_1,\bm{a}_2\cdot \bm{x}\oplus b_2,\dots,\bm{a}_m\cdot \bm{x}\oplus  b_m)) = \sum_{S\subseteq[n],|S|\leq d}c_Sy_S  = 0.
$
\par
The linear Boolean system generated through repeated queries consistently admits a valid solution originating from the true secret vector $\bm{s}$ and satisfying $y_S=\prod_{i \in S}s_i$. Arora and Ge demonstrated that, given ample samples, the components of this solution corresponding to the original variables match $\bm{s}$ with high probability.
\par
Under the adversarial noise setting involving the oracle $Q_2$, the attacker constructs a multilinear polynomial $R(\bm{\eta})$ over $\mathrm{GF}(2)$ such that $R(\bm{\eta})=0$ holds if and only if $\bm{\eta}$ decomposes as the sum $\bm{\eta}=\bm{\alpha}\oplus \bm{\beta}$ of two components satisfying $P(\bm{\alpha})=P(\bm{\beta})=0$. After injecting an auxiliary noise vector $\bm{\beta}$ satisfying $P(\bm{\beta})=0$ into the oracle samples, the combined noise incorporating the oracle noise $\bm{\alpha}$ naturally satisfies $R(\bm{\alpha}\oplus \bm{\beta})=0$. Applying the identical linearization technique yields a linear constraint over the variable vector $\bm{y}$:
$
L(R(\bm{a}_1\cdot \bm{x}\oplus b_1\oplus \beta_1,\bm{a}_2\cdot \bm{x}\oplus b_2\oplus \beta_2, \dots,\bm{a}_m\cdot \bm{x}\oplus b_m\oplus \beta_m))=0.
$
\par
The degree $d'$ of the polynomial $R$ satisfies $1 \leq d' \leq m$ and is strictly determined by the algebraic structure of $P$. This redefines the number of new variables to $N'=\sum_{i=1}^{d'}\binom{n}{i}$. By systematically acquiring ample samples, the attacker constructs a linearized Boolean system containing the genuine solution to successfully extract the secret vector.

\subsection{Framework of Quantum Algebraic Attacks}
\label{section2-3}
\subsubsection{Representation of Boolean Systems\\ over the Complex Field}
\label{section2-3-1}
To apply the HHL algorithm, the Boolean system $\mathcal{F}$ over $\mathbb{F}_2$ must be transformed into an equivalent system over the complex field $\mathbb{C}$. The mapping $C(f_i)=\prod_{k=f_i(\bm{0})}^{\lfloor t_i/2\rfloor}(f_i-2k)$ formulated in~\cite{chen2022quantum} ensures that the Boolean solution set over $\mathbb{C}$, denoted as $\mathbb{V}_B(C(\mathcal{F}))$, is strictly equivalent to the original solution set $\mathbb{V}_{\mathbb{F}_2}(\mathcal{F})$. Because the number of terms in $C(f_i)$ grows exponentially with the original term count $t_i$ and the complexity of the HHL algorithm is highly dependent on matrix sparsity, sparsifying the original system is indispensable.
\par
By introducing auxiliary variables $U_f$, each polynomial $f_i$ is decomposed into an $s$-sparse system $S(f,s)$ satisfying $s\geq 3$. Literature~\cite{chen2022quantum} proves the algebraic equivalence of this decomposition. The derived equivalent complex system $P(\mathcal{F},s)=C(S(\mathcal{F},s))$ strictly satisfies the projection relationship $\mathbb{V}_{\mathbb{F}_2}(\mathcal{F})=\operatorname{Proj}_{\mathbb{X}}\mathbb{V}_B(P(\mathcal{F},s))$ while effectively bounding the growth of the total number of variables and equations. In practice, we typically set $s=3$. Furthermore, the optimized mapping $\widehat{C}(f)$ defined by exploiting the idempotent properties of Boolean variables yields a practical number of terms substantially lower than the theoretical bound.
\subsubsection{Macaulay Linear Systems}
\label{section2-3-2}
For a complex polynomial system $\mathcal{F} = \{f_1, \dots, f_r\} \subset \mathbb{C}[\mathbb{X}]$, let the variable set $\mathbb{X} = \{x_1, \dots, x_n\}$ be lexicographically ordered, and define $\mathfrak{m}_{d}$ as the ordered set of all monomials dividing $x_1^d \cdots x_n^d$ with a cardinality of $(d+1)^n$. To normalize the system, we assume the constant terms of the first $\rho$ polynomials are $-1$ while the remainder are $0$. We select a target degree $D \geq \max_i \deg(f_i)$ and determine the smallest integers $\bar{d}$ and $\bar{D}$ satisfying $\bar{d} \geq D - \min_i \deg(f_i)$ and $\bar{D} \geq D$ such that $\bar{d}+1=2^\delta$ and $\bar{D}+1=2^\Delta$ hold for $\delta, \Delta \in \mathbb{N}$.
\par
Multiplying each $f_i$ by monomials $m_{\bar{d}, j} \in \mathfrak{m}_{\bar{d}}$ satisfying $\deg(m_{\bar{d}, j}) \leq D - \deg(f_i)$ transforms the equations into the matrix form $\mathcal{M}_{\mathcal{F},D} \bm{m}_D = \bm{b}_{\mathcal{F},D}$. Herein, the Macaulay matrix $\mathcal{M}_{\mathcal{F},D}$ has dimensions $r2^{n\delta} \times (2^{n\Delta}-1)$, the column vector $\bm{m}_D$ contains all monomials in $\mathfrak{m}_{\bar{D}} \setminus \{1\}$, and the right-hand side vector $\bm{b}_{\mathcal{F},D}$ consists of the negated constant terms $-f_i(0)$.
\subsubsection{Boolean Macaulay Linear Systems}
\label{section2-3-3}
Define the mapping $\varphi\left(\prod_{i=1}^n x_i^{a_i}\right)=\prod_{i=1}^n x_i^{\min\{1,a_i\}}$, which is linearly extended to $\mathbb{C}[\mathbb{X}]$ to convert complex polynomials into multilinear polynomials. For a polynomial system $\mathcal{F}=\{f_1,\dots,f_r\}\subseteq\mathbb{C}[\mathbb{X}]$, its degree-$d$ ($d \leq n$) Boolean Macaulay matrix $\mathcal{B}_{\mathcal{F},d}$ is generated by polynomials $\varphi(mf_i)$, where $m$ is a multilinear monomial satisfying $\deg(\varphi(mf_i)) \leq d$.
\par
The columns of $\mathcal{B}_{\mathcal{F},d}$ are indexed by all multilinear monomials of degree at most $d$ under a specified ordering. The matrix entry at row $mf_i$ and column $m'$ is precisely the coefficient of $m'$ in $\varphi(mf_i)$. Let $\bm{m}_d$ be the column vector of all such non-constant multilinear monomials, and $\bm{b}_{\mathcal{F},d}$ be the column vector comprising the negated constant terms of $\varphi(mf_i)$. This establishes the degree-$d$ Boolean Macaulay linear system: $\mathcal{B}_{\mathcal{F},d}\bm{m}_d=\bm{b}_{\mathcal{F},d}$. Since an $n$-variable set generates a maximum of $\sum_{i=0}^d \binom{n}{i}$ multilinear monomials of degree at most $d$, the dimensions of $\mathcal{B}_{\mathcal{F},d}$ are strictly bounded to $r\sum_{i=0}^d \binom{n}{i} \times \left(\sum_{i=0}^d \binom{n}{i}-1\right)$.

\section{Probabilistic Formalization of the Learning Parities with Structured Noise Problem}
\label{section3}
\begin{definition}[LPN Oracle]Let the secret vector $\bm{s} \in \mathbb{Z}^n_2$ be chosen uniformly at random, and let $p \in (0, 1)$ denote the noise rate characterizing the Bernoulli distribution $\text{Ber}(p)$. The LPN oracle $Q(n,p,\bm{s})$ outputs independent random samples $\bm{\psi} = (\bm{a}^T,b) \in \mathbb{Z}^{n+1}_2$ according to the distribution $\mathcal{Q}(n,p,\bm{s})$, defined as follows:
\begin{equation}
	\nonumber
	\begin{gathered}
		\left\{\bm{\psi} = (\bm{a}^T,b) \mid \bm{a} \xleftarrow{\text{Ber}(\frac{1}{2})} \mathbb{Z}_2^n, \eta \xleftarrow{\text{Ber}(p)} \mathbb{Z}_2, \right. \\
		\left. \vphantom{\xleftarrow{\text{Ber}(\frac{1}{2})}} \bm{s} \in \mathbb{Z}_2^n, b = \langle\bm{a}, \bm{s}\rangle \oplus \eta\right\}.
	\end{gathered}
\end{equation}
\end{definition}
Given the parameters $n$ and $p$, the attacker's objective is to recover the secret vector $\bm{s}$ by repeatedly querying this oracle.
\par
\begin{definition}[$k$-LPN Oracle]Under identical parameter settings, a single query to the $k$-LPN oracle $Q(n,k,p,\bm{s})$ is strictly equivalent to $k$ independent queries to the standard LPN oracle. It outputs samples $\bm{\psi}'=(A^T,\bm{b})$ according to the distribution $\mathcal{Q}(n,k,p,\bm{s})$, defined as:
\begin{equation}
	\nonumber
	\begin{gathered}
		\left\{ \bm{\psi}' = (A^T,\bm{b}) \mid \forall 1\leq i \leq k, \bm{a}_i \xleftarrow{\text{Ber}(\frac{1}{2})} \mathbb{Z}_2^n, \right. \\
		\left. \vphantom{\xleftarrow{\text{Ber}(\frac{1}{2})}} \eta_i \xleftarrow{\text{Ber}(p)} \mathbb{Z}_2, \bm{s}\in \mathbb{Z}_2^n, A =(\bm{a}_1,\bm{a}_2,\dots,\bm{a}_k), \right. \\
		\left. \vphantom{\xleftarrow{\text{Ber}(\frac{1}{2})}} \bm{\eta}=(\eta_1,\eta_2,\dots,\eta_k)^T, \bm{b} = A^T\bm{s} \oplus \bm{\eta} \right\}.
	\end{gathered}
\end{equation}
\end{definition}
Specifically, the oracle generates a random matrix $A_{n\times k}$ and a noise vector $\bm{\eta} \in \mathbb{Z}^k_2$ to compute $\bm{b} = A^T\bm{s} \oplus \bm{\eta}$. The attacker, aware of $n, k$, and $p$, aims to recover $\bm{s}$ from these outputs. Building upon this theoretical foundation, we subsequently introduce the LPSN oracle to formalize the correlations among the sample noise components.
\par 
\begin{definition}[Randomized LPSN Oracle]Let the secret vector $\bm{s} \in \mathbb{Z}^n_2$ be chosen uniformly at random. Given an $m$-variable, non-zero multilinear polynomial $P(\bm{\eta})$ of degree $d$, let $\mu$ be an arbitrary distribution over all valid noise patterns satisfying $P(\bm{\eta}) = 0$. The randomized LPSN oracle $Q_1(n,m,P,\mu,\bm{s})$ outputs independent samples $\bm{\psi}_1=(A^T,\bm{b})$ according to the distribution $\mathcal{Q}_1$, defined as follows:
\begin{equation}
	\nonumber
	\begin{gathered}
		\left\{ \vphantom{\xleftarrow{\text{Ber}(\frac{1}{2})}}
		\bm{\psi}_1 = ((\bm{a}_1^T,b_1)^T,(\bm{a}_2^T,b_2)^T,\dots,(\bm{a}_m^T,b_m)^T)^T= \right. \\
		\left. \vphantom{\xleftarrow{\text{Ber}(\frac{1}{2})}}
		(A^T,\bm{b}) \mid \forall 1\leq i \leq m, \bm{a}_i \xleftarrow{\text{Ber}(\frac{1}{2})} \mathbb{Z}_2^n, \bm{\eta} \xleftarrow{\mu} \mathbb{Z}_2^m, \bm{s}\in \mathbb{Z}_2^n, \right. \\
		\left. \vphantom{\xleftarrow{\text{Ber}(\frac{1}{2})}}
		A =(\bm{a}_1,\bm{a}_2,\dots,\bm{a}_m), \mu \in \{ D \in \mathcal{P}(\mathbb{Z}_2^m) \mid\right. \\
		\left. \vphantom{\xleftarrow{\text{Ber}(\frac{1}{2})}}  \operatorname{supp}(D) = \{ \bm{\eta} \mid P(\bm{\eta}) = 0 \} \},  \bm{b} = A^T\bm{s} \oplus \bm{\eta} \right\}.
	\end{gathered}
\end{equation}
\end{definition}
In this model, the oracle uniformly samples the matrix $A_{n\times m}$ and independently samples the noise vector $\bm{\eta} \sim \mu$ (i.e., it cannot observe $A$ when selecting $\bm{\eta}$), outputting $\bm{b} = A^T\bm{s} \oplus \bm{\eta}$. The attacker, aware of $n, m,$ and $P$ but ignorant of the specific distribution $\mu$, aims to recover $\bm{s}$ through oracle queries.
\par
\begin{definition}[Adversarial LPSN Oracle]Under identical parameter settings, the adversarial LPSN oracle $Q_2(n,m,P,\mu,\bm{s})$ allows for adaptive noise selection. It outputs samples $\bm{\psi}_2=(A^T,\bm{b})$ according to the distribution $\mathcal{Q}_2$, defined as:
\begin{equation}
	\nonumber
	\begin{gathered}
		\left\{ \vphantom{\xleftarrow{\text{Ber}(\frac{1}{2})}} \bm{\psi}_2 = ((\bm{a}_1^T,b_1)^T,(\bm{a}_2^T,b_2)^T,\dots,(\bm{a}_m^T,b_m)^T)^T =  \right. \\
		\left. \vphantom{\xleftarrow{\text{Ber}(\frac{1}{2})}}
		(A^T,\bm{b}) \mid \forall 1\leq i \leq m, \bm{a}_i \xleftarrow{\text{Ber}(\frac{1}{2})} \mathbb{Z}_2^n, \bm{\eta} \xleftarrow{\mu,\bm{a}_1,\bm{a}_2,\dots,\bm{a}_m} \mathbb{Z}_2^m,  \right. \\
		\left. \vphantom{\xleftarrow{\text{Ber}(\frac{1}{2})}}  
		\bm{s}\in \mathbb{Z}_2^n, A =(\bm{a}_1,\bm{a}_2,\dots,\bm{a}_m), \mu \in \{ D \in \mathcal{P}(\mathbb{Z}_2^m) \mid  
		\right. \\
		\left. \vphantom{\xleftarrow{\text{Ber}(\frac{1}{2})}} \operatorname{supp}(D) = \{ \bm{\eta} \mid P(\bm{\eta}) = 0 \} \}, \bm{b} = A^T\bm{s} \oplus \bm{\eta} \right\}.
	\end{gathered}
\end{equation}
\end{definition}
The fundamental distinction from the randomized model is that $Q_2$ adaptively selects the valid noise vector $\bm{\eta} \sim \mu$ based on the already generated matrix $A$.
\par
Feldman demonstrated that a $T$-time algorithm for the randomized LPN problem implies a $T^3$-time algorithm for the adversarial variant~\cite{feldman2006new}. However, since Arora and Ge have established polynomial-time algorithms for LPSN, we directly evaluate specific algorithms based on their polynomial degrees under distinct noise patterns, rather than employing this generic reduction.
\par
Furthermore, the correlation polynomial $P$ must satisfy a specific algebraic condition: there must exist at least one $\bm{\rho} \in \mathrm{GF}(2)^m$ such that $\bm{\rho} \neq \bm{\eta} \oplus \bm{\eta}'$ for all valid noise pairs satisfying $P(\bm{\eta}) = P(\bm{\eta}') = 0$. Arora and Ge proved that violating this condition empowers an adaptive oracle to perfectly deceive any learning algorithm~\cite{arora2011new}. Consequently, imposing this algebraic condition is both necessary and practically reasonable.
\par

\section{Refined Complexity Characterization of Algebraic Attacks on the LPSN Problem}
\label{section4}
Existing work~\cite{arora2011new} provides only rough asymptotic estimates for classical algebraic attacks. This section strictly establishes their time complexity upper bounds by explicitly determining the sample complexity required to guarantee a specific success probability of $1-\varepsilon$.

\subsection{Precise Complexity Analysis of the\\ Bit-by-Bit Guessing Algorithm}
\label{section4-1}
We first evaluate the Bit-by-Bit guessing (B-BG) algorithm. The following theorems delineate its sample and time complexities for a given success probability.
\begin{algorithm}[H]
	\caption{B-BG}
	\label{alg:B-BG}
	\begin{algorithmic}[1]
		\REQUIRE $\{(A^T,\bm{b})\}^M, P(\bm{\eta})$.
		\ENSURE $\bm{x} \in \mathbb{Z}_2^n$.
		\FOR{$j = 1$ \textbf{to} $n$}
		\FOR{$l = 1$ \textbf{to} $M$}
		\FOR{$i = 1$ \textbf{to} $m$}
		\STATE $\hat{a}_{i,j} \xleftarrow{\text{Ber}(\frac{1}{2})} \mathbb{Z}_2$
		\STATE $\hat{\bm{a}}_{i} \leftarrow (a_{i,1},\dots,\hat{a}_{i,j},\dots,a_{i,n})$
		\ENDFOR
		\STATE $L_l(P(\hat{\bm{a}}_1\cdot \bm{x}\oplus b_1,\hat{\bm{a}}_2\cdot \bm{x}\oplus b_2,\dots,\hat{\bm{a}}_m\cdot \bm{x}\oplus b_m))=0$
		\ENDFOR
		\STATE $\mathcal{L} \leftarrow \{L_k=0 \mid 1\leq k \leq M\}$
		\IF{$\text{Sol}(\mathcal{L}) = \emptyset$}
		\STATE $x_j \leftarrow 1$
		\ELSE
		\STATE $x_j \leftarrow 0$
		\ENDIF
		\ENDFOR
		\RETURN $\bm{x} = (x_1, x_2, \dots, x_n)$
	\end{algorithmic}
\end{algorithm}
\par
\begin{theorem}By making exactly $(N+\log_2\frac{1}{\varepsilon})2^d$ queries to the oracle $Q_1(n,m,P,\mu,\bm{s})$, Algorithm 1 recovers the secret vector with probability at least $1-\varepsilon$.
\end{theorem}
\par
\begin{theorem}Given $(N+\log_2\frac{1}{\varepsilon})2^d$ samples from the oracle $Q_1(n,m,P,\mu,\bm{s})$, the time complexity of Algorithm 1 is strictly bounded by $\mathcal{O}((n^d+\log_2\frac{1}{\varepsilon})2^d n^{2d+1})$.
\end{theorem}
\subsection{Complexity Evaluation of the Algebraic Algorithm Based on Gaussian Elimination}
\label{section4-2}
A primary limitation of the B-BG algorithm is its strict requirement for independence between the noise vector $\bm{\eta}$ and the uniformly random vectors $\bm{a}_i$. To address this, we introduce an algebraic algorithm based on Gaussian elimination (A-GE), which is applicable to adversarial structured noise patterns under specific algebraic conditions.
\subsubsection{Randomized Structured Noise Patterns}
\label{section4-2-1}
The following theorems establish the requisite complexities for a fixed success probability under randomized noise.

\begin{algorithm}[H]
	\caption{A-GE}
	\label{alg:A-GE}
	\begin{algorithmic}[1]
		\REQUIRE $\{(A^T,\bm{b})\}^{M'}, P(\bm{\eta})$.
		\ENSURE $\bm{x} \in \mathbb{Z}_2^n$.
		\FOR{$l = 1$ \textbf{to} $M'$}
		\STATE $L_l(P(\bm{a}_1\cdot \bm{x}\oplus b_1, \bm{a}_2\cdot \bm{x}\oplus b_2, \dots, \bm{a}_m\cdot \bm{x}\oplus  b_m))=0$
		\ENDFOR
		\STATE $\mathcal{L} \leftarrow \{L_k=0 \mid 1 \leq k \leq M'\}$ 
		\STATE $\bm{x} \leftarrow \text{Sol}(\mathcal{L})$
		\RETURN $\bm{x} = (x_1, x_2, \dots, x_n)$
	\end{algorithmic}
\end{algorithm}
\par 
\begin{theorem}By making exactly $(N+\log_2\frac{1}{\varepsilon})2^{m+d}$ queries to the oracle $Q_1(n,m,P,\mu,\bm{s})$, Algorithm 2 recovers the secret vector with probability at least $1-\varepsilon$.
\end{theorem}
\par
\begin{theorem}Given $(N+\log_2\frac{1}{\varepsilon})2^{m+d}$ samples from the oracle $Q_1(n,m,P,\mu,\bm{s})$, the time complexity of Algorithm 2 is strictly bounded by $\mathcal{O}((n^d+\log_2\frac{1}{\varepsilon})2^{m+d}n^{2d})$.
\end{theorem}
\subsubsection{Adversarial Structured Noise Patterns}
\label{section4-2-2}
Under the adversarial noise setting, the attacker can substitute the correlation polynomial $P$ with a constructed polynomial $R$ to solve the resulting linear system via Algorithm 2. The complexities for achieving a specified success probability are outlined below.
\par
\begin{theorem}By making exactly $(N'+\log_2\frac{1}{\varepsilon})2^{d'}$ queries to the oracle $Q_2(n,m,P,\mu,\bm{s})$, Algorithm 2 recovers the secret vector with probability at least $1-\varepsilon$.
\end{theorem}
\par
\begin{theorem}Given $(N'+\log_2\frac{1}{\varepsilon})2^{d'}$ samples from the oracle $Q_2(n,m,P,\mu,\bm{s})$, the time complexity of Algorithm 2 is strictly bounded by $\mathcal{O}((n^{d'}+\log_2\frac{1}{\varepsilon})2^{m+d'}n^{2d'})$.
\end{theorem}
\section{Complexity Bounds of Reduction-Optimized Quantum Algebraic Attacks on the LPSN Problem}
\label{section5}
In this section, we propose a novel reduction method to optimize the quantum algorithm for solving Boolean systems, and subsequently apply it to the LPSN problem.

\subsection{Sample Complexity for Quantum\\ Algebraic Systems}
\label{section5-1}
For solving non-linear Boolean systems, we employ the methodology from Section \ref{section2-3-1}. By introducing auxiliary variables, the system is transformed into an equivalent 3-sparse Boolean system, and subsequently into a 6-sparse complex-field system. Field equations are appended to ensure strict equivalence between the complex-field solution set and the original Boolean solution set. Finally, the quantum algorithm evaluates the corresponding Macaulay linear system.
\par
\begin{theorem}Suppose the oracle $Q_2(n,m,P,\mu,\bm{s})$ uniformly selects a secret vector $\bm{s}\in \mathrm{GF}(2)^n$ and vectors $\bm{a}_i\in \mathrm{GF}(2)^n$, and adaptively samples a noise vector $\bm{\eta} \sim \mu$ satisfying $P(\bm{\eta})=0$, where $P(\bm{\eta})$ is a degree-$d$ non-linear polynomial satisfying the requisite algebraic condition. Given $b_i=\bm{a}_i\cdot \bm{s}\oplus \eta_i$, we construct non-linear Boolean equations over $\bm{x}$: $P(\bm{a}_1\cdot \bm{x}\oplus b_1,\dots,\bm{a}_m\cdot \bm{x}\oplus b_m)=0$. When the number of equations reaches $\lceil 2^m\ln\frac{1}{\varepsilon}\rceil$, this system admits a unique solution $\bm{x}=\bm{s}$ with a probability of at least $1-\varepsilon$.
\end{theorem}
\begin{proof}
	Clearly, $\bm{x}=\bm{s}$ is a valid solution. Assume another solution $\bm{s}'\neq \bm{s}$ exists, defining a non-zero difference vector $\bm{v}=\bm{s}'\oplus\bm{s}$. Since $\bm{s}'$ satisfies the equations, we have $\bm{a}_i\cdot \bm{s}'\oplus b_i = \bm{a}_i\cdot \bm{v}\oplus \eta_i$. Let $\delta_i=\bm{a}_i\cdot \bm{v}$; thus, $P(\bm{\delta}\oplus \bm{\eta})=0$ holds. Because both $\bm{\eta}$ and $\bm{\delta}\oplus \bm{\eta}$ are valid noise patterns, the algebraic condition ($\bm{\eta}\oplus \bm{\eta}'\neq \bm{\rho}$) necessitates that $\bm{\delta}\neq\bm{\rho}$.
	\par
	Although $\bm{\eta}$ is chosen adaptively, this occurs after the random selection of $\bm{a}_i$. Since $\bm{\delta}$ depends strictly on $\bm{a}_i$ and the fixed $\bm{v}$, it is completely determined. Given uniformly random $\bm{a}_i$ and $\bm{v}\neq \bm{0}$, $\bm{\delta}$ is uniformly distributed in $\mathrm{GF}(2)^m$. The probability that $\bm{\delta}=\bm{\rho}$ for a single sample is $\frac{1}{2^m}$. For $T$ equations, the probability that all samples satisfy $\bm{\delta}\neq\bm{\rho}$ is $(1-\frac{1}{2^m})^T$. Setting $(1-\frac{1}{2^m})^T\leq \varepsilon$ and applying $\ln(1-p)< -p$ yields $T\geq\lceil 2^m\ln\frac{1}{\varepsilon}\rceil$. Thus, with this sample size, the system admits a unique solution with a probability of at least $1-\varepsilon$.
\end{proof}
In classical algebraic attacks, direct linearization cannot guarantee that the solution of the linear system satisfies the inherent non-linear constraints, necessitating massive samples to probabilistically isolate the secret vector's components. Conversely, as established in Section \ref{section2-3-2}, the quantum algorithm leverages the solution structure of the Macaulay linear system to successfully recover the secret vector simply by ensuring the unique solvability of the original non-linear system, thereby exponentially reducing the sample complexity.

\subsection{A Novel Polynomial Reduction Method}
\label{section5-2}
Aharonov and Ta-Shma~\cite{aharonov2003adiabatic} proved that, in general, there exists no efficient method for preparing an arbitrary given quantum state. Ding et al.~\cite{ding2023limitations} proposed a reduction method, denoted as Red1, which enables the efficient preparation of the requisite quantum state. To advance this, we propose a novel reduction method, Red2, which guarantees efficient state preparation while concurrently lowering the condition number lower bound.
\par
Let $\mathbb{X}=\{x_1,x_2,\dots,x_n\}$ be a variable set and $\mathcal{F}=\{f_1,f_2,\dots ,f_r\}\subseteq\mathbb{C}[\mathbb{X}]$ be a polynomial system. Solving a Boolean system over the complex field necessitates the field equations $\mathbb{H}_{\mathbb{X}}=\{ x_i^2-x_i\}_{i=1}^n$. To mitigate the impact of the reduction on the Macaulay matrix's spectral norm, we directly fine-tune the Macaulay matrix generated by $\mathcal{F}\cup \mathbb{H}_{\mathbb{X}}$. Specifically, we exclusively apply the Red2 to transform $\mathcal{F}$ into $\mathcal{F}''=\{f_1'',f_2'',\dots,f_r''\}$ and substitute the $r$ rows generated by $f_1,f_2,\dots,f_r$, leaving the rows generated by $\mathbb{H}_{\mathbb{X}}$ completely unaltered.
\par
Assume $\mathcal{F}$ contains at least one polynomial with a non-zero constant term; if all lack constant terms, $\bm{0}$ is trivially a Boolean solution. Let $c_i$ denote the constant term of $f_i$, and without loss of generality, assume $c_1\neq 0$. The Red1 method proposed by Ding et al. and our novel Red2 method are detailed as follows:
\par 
\textbf{Red1:} Let $f_1'=-\frac{f_1}{c_1}$. For all $i\in \{2,3,\dots,r\}$, define $f_i'=f_i+c_if_1'$. This procedure sets the constant term of $f_1'$ to $-1$ and eliminates the constant terms in all other polynomials.
\par
\textbf{Red2:} Let $f_1''=-\frac{f_1}{c_1}$. For all $i\in \{2,3,\dots,r\}$, define the following piecewise function to strictly reduce all constant terms to $-1$:
$$
f_i'' = 
\begin{cases}
	-\frac{f_i}{c_i}, & \text{if } f_i \text{ has a constant term } c_i; \\
	f_i+f_1'', & \text{if } f_i \text{ has no constant term}.
\end{cases}
$$
\par
As is evident from the definitions, these modifications to the $r$ rows of the Macaulay matrix are strictly equivalent to applying a sequence of elementary row operations to both the matrix and its corresponding right-hand side vector. Consequently, the solution space of the Macaulay linear system remains invariant.
\par
Harrow et al.~\cite{harrow2009quantum} did not account for the preparation time of the quantum state $|\bm{b}\rangle$ when utilizing the HHL algorithm to solve the linear system $A\bm{x}=\bm{b}$. Chen and Gao~\cite{chen2022quantum} proposed a state preparation method applicable to vectors $\bm{b}$ satisfying specific conditions. Our proposed Red2 ensures that the right-hand side vector of the adjusted Macaulay linear system precisely fulfills these conditions.
Appending the equations $\bm{0} \cdot \bm{x}=1$ and $\bm{0} \cdot \bm{x}=0$ to a linear system perfectly preserves the quantum solution state output by the HHL algorithm, which guarantees the validity of Theorem V.2.
\par
\begin{theorem}For a linear system $A\bm{x}=\bm{b}$, if the matrix $A\in \mathbb{C}^{M\times N}$ and the vector $\bm{b}\in\{0,1\}^M$ satisfy $M=m2^v$ and the component $b_i=1$ if and only if $i=k2^v$ for $k=0,1,\dots,\rho-1$, then the requisite quantum state for the HHL algorithm can be prepared in $\mathcal{O}(\log_2(M+N))$ time.
\end{theorem}
\par
The right-hand side vector $\bm{b}_R$ in the new Macaulay linear system $\mathcal{M}_R\bm{m}_D=\bm{b}_R$ strictly satisfies the prerequisites of Theorem V.2. Because the system transformations preserve the $s$-sparsity of the matrix, if the $1$-sparse submatrices decomposed from the original matrix exhibit a query complexity of $\mathcal{O}(\gamma)$, the time required for the HHL algorithm to output an $\varepsilon'$-approximate solution state is strictly bounded by $\mathcal{O}((\log_2(M+N)+\gamma)s\kappa_{\bm{b}}(A)^2/\varepsilon')$. 
Chen and Gao demonstrated that Macaulay matrices inherently satisfy the condition of $\gamma$ being sufficiently small, thereby simplifying the ultimate time complexity for solving this system to $\mathcal{O}((\log_2(M+N))s\kappa^2/\varepsilon')$, where $\kappa=\kappa_{\bm{b}}(A)$.

\par 
\subsection{A Quantum Algorithm for Reduced\\ Macaulay Linear Systems}
\label{section5-3}
\subsubsection{Algorithmic Framework and Solution Procedure}
\label{section5-3-1}
Based on the theoretical framework established above, we propose an efficient quantum algorithm for finding Boolean solutions of complex polynomial systems via Macaulay matrices (Q-BSCP-M). Given that the uniqueness of the solution is guaranteed when solving the LPSN problem, we omit the generalized version designed to output all solutions.
\begin{figure*}[!t] 
	\begin{minipage}{\textwidth} 
		\let\centering\raggedright
		\begin{algorithm}[H] 
			\caption{Q-BSCP-M ($\mathcal{F},\varepsilon'$)}
			\label{alg:Q-BSCP-M}
			\begin{algorithmic}[1]
				\REQUIRE $\mathcal{F} = \{f_1, \dots, f_r\} \subseteq \mathbb{C}[\mathbb{X}]$, $\varepsilon' \in (0, 1)$.
				\ENSURE $\bm{a} \in \mathbb{V}_{B}(\mathcal{F}) \lor \emptyset$.
				
				\STATE $\mathcal{F}_{1} \leftarrow \mathcal{F} \pmod{\langle x^2 - x \mid x \in \mathbb{X} \rangle}, \quad l \leftarrow 1, \quad \varepsilon_1' \leftarrow \max\{1/2, \varepsilon'\}$
				
				\IF{$\mathcal{F}_1(\bm{1})=\bm{0}$}
				\RETURN $\bm{1}$
				\ENDIF
				
				\WHILE{$l \le \lceil \log_{\varepsilon_1} \varepsilon \rceil$}
				\STATE $\mathbb{Y} \leftarrow \mathbb{X}, \quad \mathcal{F}_1 \leftarrow \mathcal{F}_B$
				
				\WHILE{$\mathcal{F}_1 \neq \emptyset \land \mathcal{F}_1(\mathbf{0}) \neq \mathbf{0}$}
				\IF{$\mathcal{F}_1 \cap \mathbb{C} \setminus \{0\} \neq \emptyset$}
				\STATE $l \leftarrow l + 1, \quad \textbf{break}$ 
				\ENDIF
				
				\STATE $\mathcal{F}_2 \leftarrow \mathcal{F}_1 \cup \mathbb{H}_{\mathbb{Y}}, \quad D \leftarrow 3\#\mathbb{Y}, \quad (\mathcal{M}_{\mathcal{F}_2, D}, \bm{b}_{\mathcal{F}_2, D}) \leftarrow \text{ConstructSys}(\mathcal{F}_2, D)$
				\STATE $\mathcal{F}_R \leftarrow \text{Red2}(\mathcal{F}_1), \quad (\mathcal{M}_R, \bm{b}_R) \xleftarrow{ \mathcal{F}_R } \text{ReConstructSys}(\mathcal{F}_2, D), \quad |\widehat{\bm{m}}_D\rangle \leftarrow \textbf{HHL}(\mathcal{M}_{R}, \bm{b}_{R}, \sqrt{\varepsilon'_1/n})$  
				
				\STATE $|e\rangle \xleftarrow{\text{Measure}} |\widehat{\bm{m}}_D\rangle, \quad m_{\bar{D},k} = \prod_{i=1}^{u_k} x_{n_i} \leftarrow \text{Identify}(|e\rangle)$
				\STATE $x_{n_i} \leftarrow 1 \; (\forall i \in \{1, \dots, u_k\}), \quad \mathcal{F}_1 \leftarrow \mathcal{F}_1\big|_{x_{n_i}=1} \setminus \{0\}, \quad \mathbb{Y} \leftarrow \mathbb{Y} \setminus \{x_{n_i} \mid i=1,\dots,u_k\}$
				\ENDWHILE
				
				\FOR{$i = 1$ \TO $n$}
				\IF{$x_i \in \mathbb{Y}$}
				\STATE $a_i \leftarrow 0$
				\ELSE
				\STATE $a_i \leftarrow 1$
				\ENDIF
				\ENDFOR
				\RETURN $\bm{a} = (a_1, a_2, \dots, a_n)$
				\ENDWHILE
				
				\RETURN $\emptyset$
			\end{algorithmic}
		\end{algorithm}
	\end{minipage} 
\end{figure*}
\par
The core of Algorithm 3 utilizes the HHL algorithm to solve the Macaulay linear system $\mathcal{M}_{R}\bm{m}_D=\bm{b}_{R}$, where $\mathcal{M}_{R}\in \mathbb{C}^{(r+n)2^{n\delta}\times (2^{\Delta}-1)}$. In conjunction with the Macaulay matrix modification scheme detailed in Step 15, the vector $\bm{b}_{R}$ strictly satisfies the conditions in Theorem V.2 corresponding to $\rho=r$, $m=r+n$, and $v=n\delta$. Consequently, the requisite quantum state can be prepared in $\mathcal{O}(\log_2(M+N))$ time.
\par
Chen and Gao demonstrated that the degree-$D$ Macaulay matrix $\mathcal{M}_{\mathcal{F},D}$ generated by $\mathcal{F}$ is $T_{\mathcal{F}}$-sparse and inherently admits a 1-sparse decomposition in~\cite{chen2022quantum}. When employing Algorithm 3 to find Boolean solutions for $\mathcal{F}$, the initialization phase replaces all variables with degrees greater than 1 with degree-1 variables, yielding the system $\mathcal{F}_1$. Since distinct monomials may map to identical multilinear polynomials, the inequality $T_{\mathcal{F}_1} \leq T_{\mathcal{F}}$ holds. Thus, $\mathcal{M}_{\mathcal{F}_1,D}$ is at most $T_{\mathcal{F}}$-sparse, and the sparsity of $\mathcal{M}_{\mathcal{F}_2,D}$ corresponding to $\mathcal{F}_2=\mathcal{F}_1\cup \mathbb{H}_{\mathbb{Y}}$ is bounded by $T_{\mathcal{F}}+2n$. Applying the Red2 to $\mathcal{F}_1$ yields $\mathcal{F}_R$, which is subsequently utilized to modify the corresponding $r$ rows in $\mathcal{M}_{\mathcal{F}_2,D}$. The sparsity of the resulting matrix $\mathcal{M}_R$ is bounded by $2T_{\mathcal{F}}+2n$, though we omit the impact of this operation on matrix sparsity in subsequent considerations. The time complexity of Steps 14 and 15 is at most $\mathcal{O}(r)$, establishing the validity of Theorem V.3. 

\par
\begin{theorem}If $\mathbb{V}_{B}(\mathcal{F})=\emptyset$, Algorithm 3 outputs $\emptyset$; otherwise, it outputs a valid Boolean solution to the system $\mathcal{F}=0$ with a probability of at least $1-\varepsilon'$. The approximate time complexity of the algorithm is bounded by $\mathcal{O}(n^{\frac{5}{2}}(n+T_{\mathcal{F}})\kappa^2 \log_{2} \frac{1}{\varepsilon'} )$, where $\kappa=\max \kappa_{\bm{b}_{R}}(\mathcal{M}_{R})$ denotes the maximum condition number of all modified Macaulay matrices $\mathcal{M}_{R}$ utilized in Step 16 with respect to the vector $\bm{b}_{R}$.
\end{theorem}
\par

In practice, Algorithm 3 serves as a subroutine within Algorithm 4, an efficient quantum algorithm for solving Boolean polynomial systems via Macaulay matrices(Q-SBP-M). Fig.~\ref{fig1} illustrates the flowchart for solving Boolean systems via Algorithm 4, followed by its properties.

\begin{algorithm}[H]
	\caption{Q-SBP-M ($\mathcal{F}, \varepsilon'$)}
	\label{alg:Q-SBP-M}
	\begin{algorithmic}[1] 
		\REQUIRE $\mathcal{F} = \{f_1, \dots, f_r\} \subset \mathcal{R}_2[\mathbb{X}]$, $\varepsilon' \in (0, 1)$.
		\ENSURE $\bm{a} \in \mathbb{V}_{\mathbb{F}_2}(\mathcal{F}) \lor \emptyset$.
		
		\STATE $\mathbb{Y} \leftarrow \mathbb{X} \cup \mathbb{U}(\mathcal{F}, 3)$  
		\STATE $\mathcal{F}_1\leftarrow S(\mathcal{F}, 3) \subseteq \mathcal{R}_2[\mathbb{Y}]$  
		\STATE $\mathcal{F}_2 \leftarrow \widehat{C}(\mathcal{F}_1) \subseteq \mathbb{C}[\mathbb{Y}]$  
		\STATE $\bm{a} \leftarrow \text{Q-BSCP-M} (\mathcal{F}_2, \varepsilon')$  
		
		\IF{$\bm{a} \neq \emptyset$}
		\RETURN $\text{Proj}_{\mathbb{X}}\ \bm{a}$
		\ENDIF
		
		\RETURN $\emptyset$
	\end{algorithmic}
\end{algorithm}
\par

\par 
\begin{figure*}[!t]
	\centering
	\includegraphics[width=6.3in]{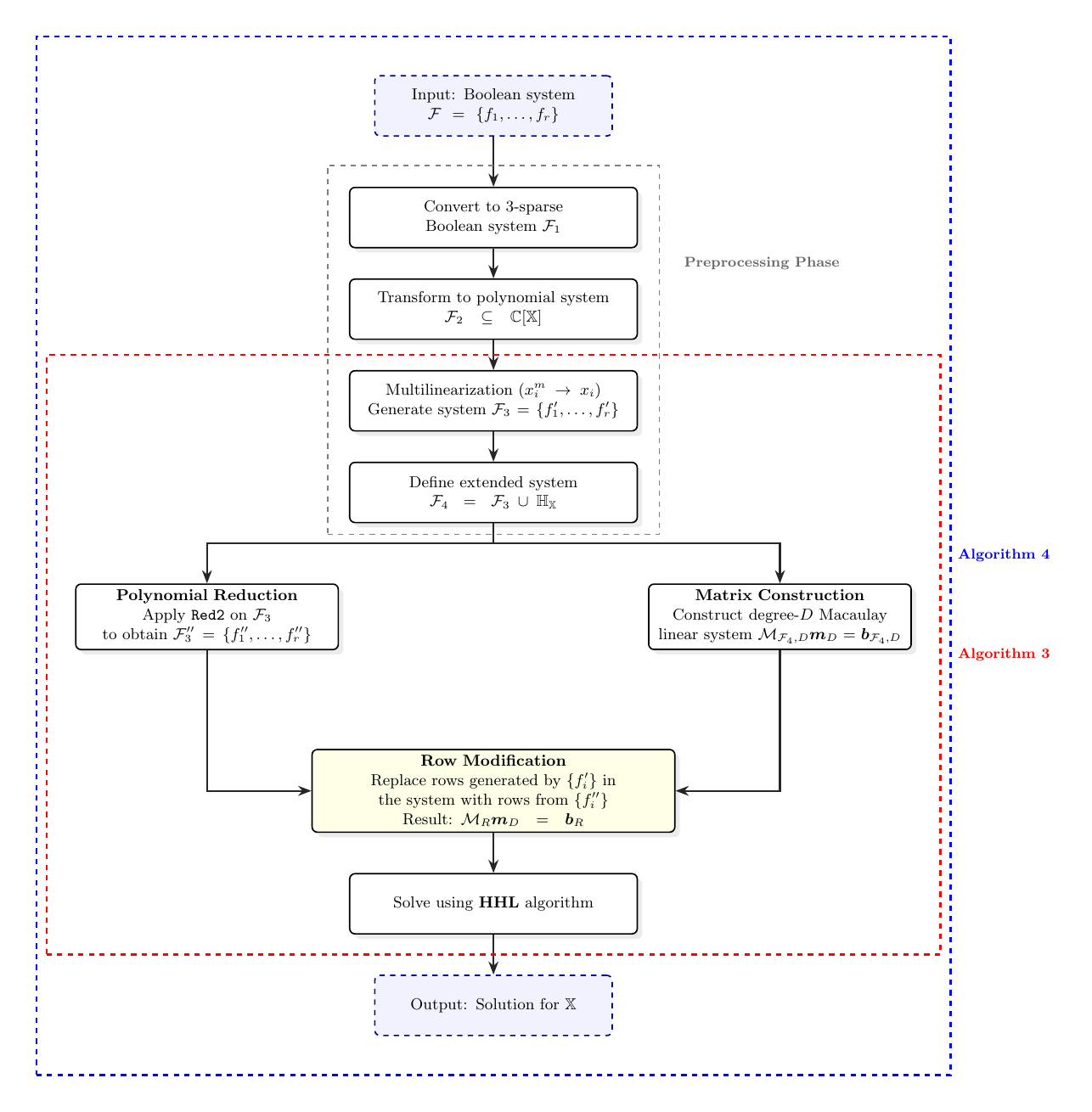}
	\caption{Macaulay-based quantum algorithm for Boolean systems.}
	\label{fig1}
\end{figure*}
\begin{theorem}If $\mathbb{V}_{\mathbb{F}_2}(\mathcal{F})=\emptyset$, Algorithm 4 outputs $\emptyset$; otherwise, it outputs a valid solution to the Boolean system $\mathcal{F}=0$ with a probability of at least $1-\varepsilon'$. The approximate time complexity of the algorithm is bounded by $\mathcal{O}((n^{\frac{7}{2}}+T_{\mathcal{F}}^{\frac{7}{2}})\kappa^2\log_2\frac{1}{\varepsilon'})$, where $\kappa=\max \kappa_{\bm{b}_{R}}(\mathcal{M}_{R})$ denotes the maximum condition number of all Macaulay matrices with respect to specific vectors used while finding Boolean solutions to the system $\mathcal{F}_2=0$ over the complex field via Algorithm 3.
\end{theorem}
\par
The complexity outlined in Theorem V.4 strictly depends on $\kappa$ and $1/\varepsilon'$. By reducing the error rate requirement to $\sqrt{\varepsilon'_1/n}$ and executing multiple iterations, Algorithm 3 exponentially reduces the dependence of the time complexity on $1/\varepsilon'$ while consistently ensuring a success probability of at least $1-\varepsilon'$.

\subsubsection{Bound Analysis of the Condition Number}
\label{section5-3-2}
First, appending $p=(2^\sigma-\rho)2^v$ zero rows and zero components to $\mathcal{M}_R \in \mathbb{C}^{M\times N}$ and $\bm{b}_R$ respectively yields $\mathcal{M}_1 \in \mathbb{C}^{(M+p)\times N}$ and $\bm{b}_1 \in \mathbb{C}^{(M+p)\times 1}$. Subsequently, appending $q=2^\eta-((r+2^\sigma-\rho)2^v+N)$ zero rows and zero components constructs $\mathcal{M}_2$ and $\bm{b}_2$. The pseudoinverse $\mathcal{M}_2^+ = \begin{pmatrix} \mathcal{M}_R^+ & \bm{0} \end{pmatrix}$ implies $\mathcal{M}_2^+\bm{b}_2 = \mathcal{M}_R^+\bm{b}_R$. Given $\|\mathcal{M}_2\| = \|\mathcal{M}_R\|$ and that the additional non-zero components in $\bm{b}_2$ compared to $\bm{b}_R$ ensure $\|\bm{b}_2\| \geq \|\bm{b}_R\|$, we deduce $\kappa_{\bm{b}_2}(\mathcal{M}_2) \leq \kappa_{\bm{b}_R}(\mathcal{M}_R)$.
\par 
To satisfy the HHL algorithm's requirement for a self-adjoint transformation, the system solved in practice is formulated as $\mathcal{M}_3\bm{y}=\bm{b}_3$, where $\mathcal{M}_3 = \begin{pmatrix} \bm{0} & \mathcal{M}_2 \\ \mathcal{M}_2^H & \bm{0} \end{pmatrix}$ and $\bm{b}_3 = \begin{pmatrix} \bm{b}_2 \\ \bm{0} \end{pmatrix}$. This formulation yields $\mathcal{M}_3^+\bm{b}_3 = \begin{pmatrix} \bm{0} \\ \mathcal{M}_2^+\bm{b}_2 \end{pmatrix}$ and $\|\mathcal{M}_3^+\bm{b}_3\| = \|\mathcal{M}_2^+\bm{b}_2\|$. Combined with $\|\mathcal{M}_3\| = \|\mathcal{M}_2\|$ and $\|\bm{b}_3\| = \|\bm{b}_2\|$, it follows that $\kappa_{\bm{b}_3}(\mathcal{M}_3) \leq \kappa_{\bm{b}_R}(\mathcal{M}_R)$.
\par 
This demonstrates that preprocessing does not degrade the linear system but rather enhances its stability. Therefore, directly evaluating the condition number upper bound of the original reduced system $\kappa_{\bm{b}_R}(\mathcal{M}_R)$ is both reasonable and sufficiently rigorous.

\par 

In Fig.~\ref{fig1}, let $\mathcal{M}\bm{m}_D=\bm{b}$ denote the Macaulay linear system corresponding to $\mathcal{F}_4$. The reduction driven by $\mathcal{F}_3''$ is equivalent to elementary row operations; thus, the reduced system $\mathcal{M}_R\bm{m}_D=\bm{b}_R$ shares the identical solution space and minimal norm solution $\mathcal{M}_R^+\bm{b}_R=\mathcal{M}^+\bm{b}$ with the original system. As detailed in Section \ref{section5-2}, $\bm{b}_R$ contains exactly $T_{\mathcal{F}}-2r$ non-zero components equal to $1$, yielding $\|\bm{b}_R\|=(T_{\mathcal{F}}-2r)^{\frac{1}{2}}$. Adopting the assumptions $\|\mathcal{M}_R\|=\|\mathcal{M}\|$ and $\|\mathcal{M}\|\geq 1$, we obtain $\kappa_{\bm{b}_R}(\mathcal{M}_R) \geq (T_{\mathcal{F}}-2r)^{-\frac{1}{2}}\|\mathcal{M}^+\bm{b}\|$. If the Boolean system possesses a unique solution $\bm{a}$ with Hamming weight $h$, the minimal norm solution output by the HHL algorithm contains $\binom{D+h}{h}-1$ non-zero components. By selecting the full solving degree $D=3n$, the lower bound of the condition number is strictly determined as $\kappa_{\bm{b}_R}(\mathcal{M}_R)\geq \left(T_{\mathcal{F}}-2r\right)^{-\frac{1}{2}}\left({\binom{3n+h}{h}-1}\right)^{\frac{1}{2}}$.
\par

\par 
Consider a 3-sparse Boolean polynomial system comprising 5 variables and 27 polynomials. Assume the system has undergone a multilinear transformation, and the corresponding Macaulay linear system is fully solvable when $D=2$, at which point $\bar{d}=1$ and $\bar{D}=3$. After appending the field polynomials, each polynomial generates 32 rows, constructing a Macaulay matrix with exact dimensions of $1024\times 1023$, where the first 864 rows originate from the complex field polynomials and the subsequent 160 rows from the field polynomials. Based on the constant term distribution and equivalent transformation rules, we designed a greedy search algorithm driven by iterative optimization and structural constraints. Through $10^6$ iterations, we successfully constructed a Macaulay matrix instance featuring a specific block-sparsity pattern and a minimized spectral norm bounded by $\|\mathcal{M}\|\approx 4.71$, visually represented in Fig.~\ref{fig2}.
\begin{figure*}[!t]
	\centering
	\subfloat[]{\includegraphics[width=3.5in]{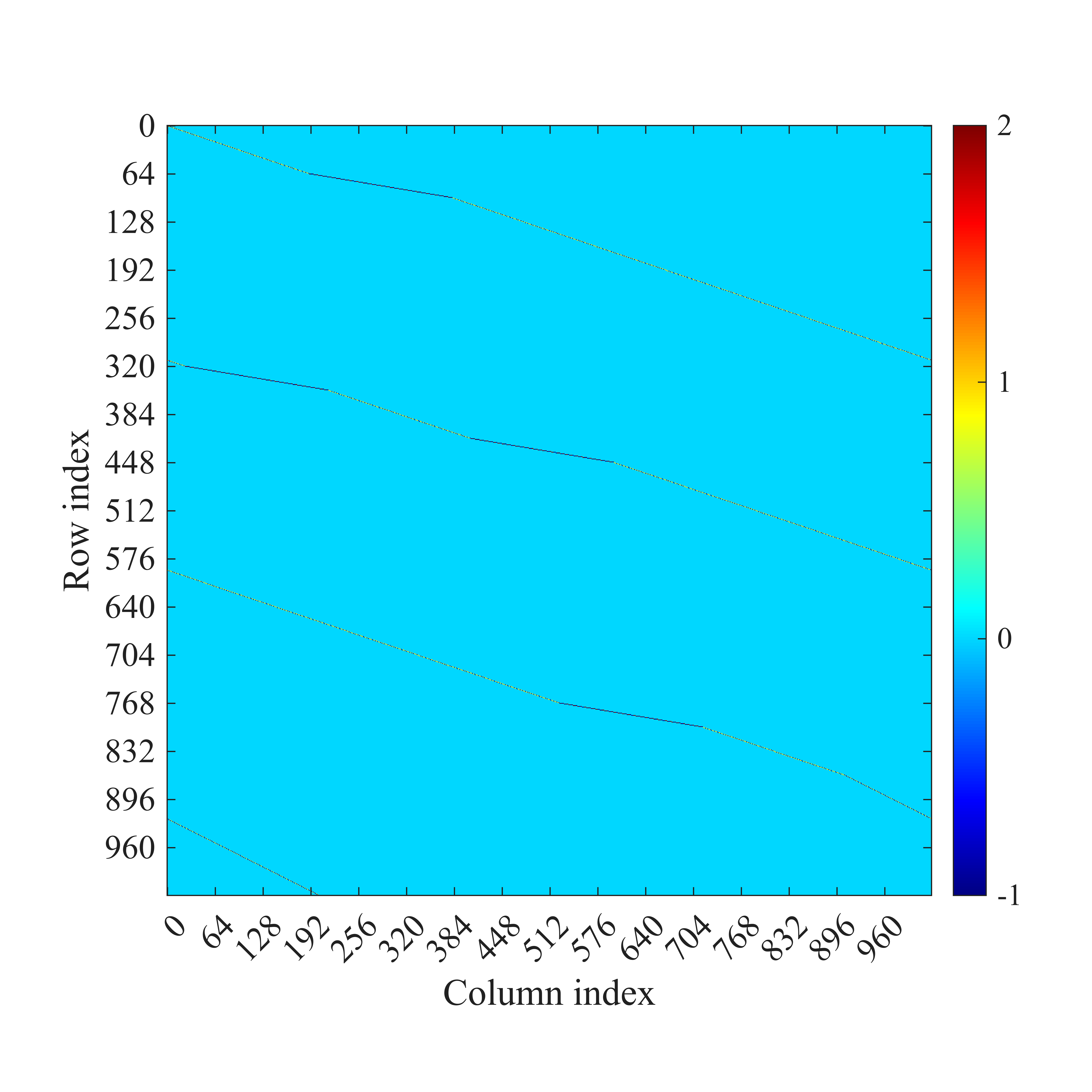}%
		\label{fig:sub1}}
	\hfil
	\subfloat[]{\includegraphics[width=3.5in]{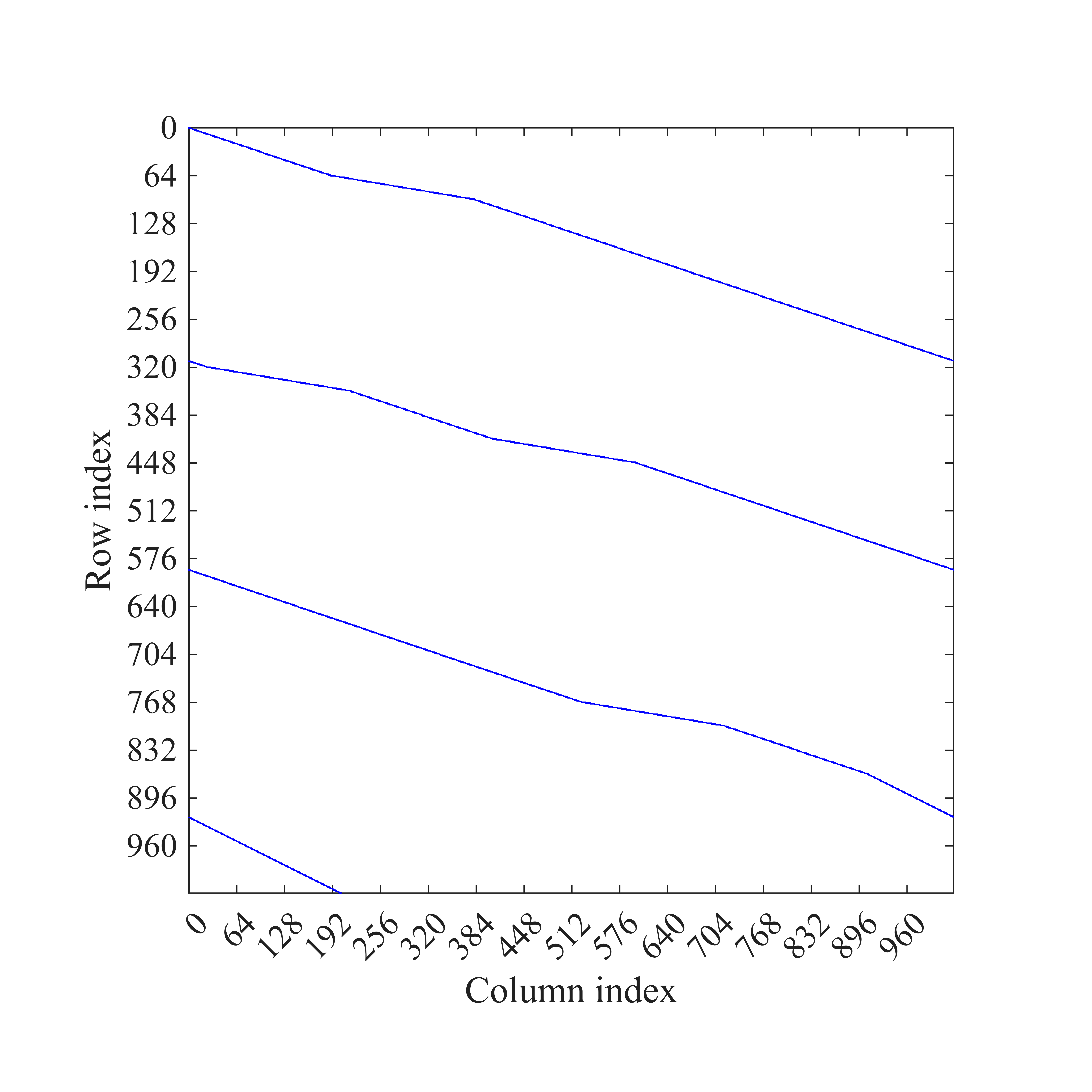}%
		\label{fig:sub2}}
	\caption{Visualization of the Macaulay matrix with dimensions $1024 \times 1023$. (a) The distribution of element values. (b) The sparsity pattern representing non-zero entries.}
	\label{fig2}
\end{figure*}
\par
The non-linear Boolean system $\mathcal{F}$ corresponding to this matrix $\mathcal{M}$ comprises 23 polynomials with constant terms and 4 without, yielding 81 terms. In this case, the reduced right-hand side vector satisfies $\|\bm{b}_{R}\|=\sqrt{27}$. Adopting Ding et al.'s assumption that $\|\mathcal{M}_R\|=\|\mathcal{M}\|$, we have $\left(T_{\mathcal{F}}-2r\right)^{-\frac{1}{2}}\left({\binom{3n+h}{h}-1}\right)^{\frac{1}{2}} < \kappa_{\bm{b}_R}(\mathcal{M}_R) < \left({\binom{3n+h}{h}-1}\right)^{\frac{1}{2}}$. This confirms that under our reduction, the condition number of the Macaulay linear system with respect to the right-hand side vector can strictly surpass the theoretical lower bound established by Ding et al.
\par
Obtaining $\mathcal{M}_R$ intrinsically equates to applying elementary row operations to the original matrix $\mathcal{M}$. Swapping two rows and dividing by the constant $1$ perfectly preserve the spectral norm. In the worst-case scenario, adding a row containing a constant term to the remaining $T_{\mathcal{F}} - 2r - 1$ rows devoid of constant terms strictly restricts the scaling of the spectral norm within the interval $[\mathcal{O}((T_{\mathcal{F}} - 2r)^{-\frac{1}{2}}), \mathcal{O}((T_{\mathcal{F}} - 2r)^{\frac{1}{2}})]$. Combined with the fact that this reduction amplifies $\|\bm{b}\|$ by a factor of $\mathcal{O}((T_{\mathcal{F}} - 2r)^{\frac{1}{2}})$, it fundamentally avoids escalating the condition number and theoretically holds the potential to decrease it.
\par
The Red1 employed by Ding et al. neglected the impact of row operations on the matrix norm. Incorporating this factor, the condition number interval under the Red1 is approximately
\begin{align*}
	& \left[ \left(T_{\mathcal{F}}-2r\right)^{-\frac{1}{2}}\left(\tbinom{3n+h}{h}-1\right)^{\frac{1}{2}}\|\mathcal{M}\|, \right. \\
	& \quad \left. \left(T_{\mathcal{F}}-2r\right)^{\frac{1}{2}}\left(\tbinom{3n+h}{h}-1\right)^{\frac{1}{2}}\|\mathcal{M}\| \right].
\end{align*}
Conversely, our Red2 strictly compresses this interval to
\begin{align*}
	& \left[ \left(T_{\mathcal{F}}-2r\right)^{-1}\left(\tbinom{3n+h}{h}-1\right)^{\frac{1}{2}}\|\mathcal{M}\|, \right. \\
	& \quad \left. \left(\tbinom{3n+h}{h}-1\right)^{\frac{1}{2}}\|\mathcal{M}\| \right].
\end{align*}
The significant reduction across both the upper and lower bounds conclusively demonstrates the superiority of our reduction method.

\subsubsection{Quantum Time Complexity Bounds for \\Solving the LPSN Problem}
\label{section5-3-3}
\par
According to Theorem V.4, the time complexity of employing Algorithm 4 to solve the system $\mathcal{F}$ is bounded by $\mathcal{O}((n^{\frac{7}{2}} + T_{\mathcal{F}}^{\frac{7}{2}})\kappa^2 \log_2 \frac{1}{\varepsilon'})$. A single sample under a degree-$d$ structured noise polynomial yields a Boolean equation with at most $\sum_{i=0}^{d} \binom{n}{i}$ terms.
Deduced from Section \ref{section5-1}, acquiring $\lceil 2^m\ln\frac{1}{\varepsilon} \rceil$ samples guarantees the unique solvability of the Boolean system with a probability of at least $1-\varepsilon$. At this point, we have $T_{\mathcal{F}} \leq \sum_{i=0}^d\binom{n}{i}\lceil 2^m\ln\frac{1}{\varepsilon} \rceil$. To achieve a success probability of at least $(1-\varepsilon)(1-\varepsilon')$, the runtime of Algorithm 4 is bounded by $\mathcal{O}\left(2^{\frac{7}{2}m}n^{\frac{7}{2}d}\kappa^2\ln^{\frac{7}{2}}\frac{1}{\varepsilon}\log_2\frac{1}{\varepsilon'}\right)$.
\par
Before any variable assignments, initially solving the Macaulay linear system $\mathcal{M}_R\bm{m}_D=\bm{b}_R$ dictates $\kappa_{\bm{b}_R}(\mathcal{M}_R)^2 \geq \left(T_{\mathcal{F}}-2r\right)^{-1}\left(\binom{3n+h}{h}-1\right)$. Since $T_{\mathcal{F}}-2r = \mathcal{O}(2^mn^d\ln\frac{1}{\varepsilon})$, and for all $h\geq 1$, the inequality $\left(\frac{3n+h}{h}\right)^h\leq \binom{3n+h}{h}\leq \left(\frac{e(3n+h)}{h}\right)^h$ holds, we deduce $\kappa_{\bm{b}_R}(\mathcal{M}_R)^2 \geq 2^{-m}\left( \frac{3}{h}\right)^hn^{h-d}\ln^{-1}\frac{1}{\varepsilon}$. Therefore, utilizing Algorithm 4 recovers the secret vector in the LPSN problem with a probability of at least $(1-\varepsilon)(1-\varepsilon')$ in a time complexity bounded by $\mathcal{O}\left(\Omega\left(2^{\frac{5}{2}m}h^{-h}n^{\frac{5}{2}d+h}\ln^{\frac{5}{2}}\frac{1}{\varepsilon}\log_2\frac{1}{\varepsilon'}\right)\right)$.

\subsection{A Quantum Algorithm for Reduced Boolean Macaulay Linear Systems}
\label{section5-4}
\subsubsection{Algorithmic Framework and Solution Procedure}
\label{section5-4-1}
By replacing the Macaulay matrix utilized in Algorithm 3 with the Boolean Macaulay matrix, we obtain an efficient quantum algorithm for finding Boolean solutions of complex polynomial systems via Boolean Macaulay matrices (Q-BSCP-BM), and prove its properties.

\begin{figure*}[!t] 
	\begin{minipage}{\textwidth} 
		\let\centering\raggedright
		\begin{algorithm}[H] 
			\caption{Q-BSCP-BM ($\mathcal{F},\varepsilon'$)}
			\label{alg:Q-BSCP-BM}
			\begin{algorithmic}[1]
				\REQUIRE $\mathcal{F} = \{f_1, \dots, f_r\} \subseteq \mathbb{C}[\mathbb{X}]$; $\varepsilon' \in (0, 1)$.
				\ENSURE $\mathbf{a} \in \mathbb{V}_{B}(\mathcal{F}) \lor \emptyset$.
				
				\STATE $\mathcal{F}_{1} \leftarrow \mathcal{F} \pmod{\langle x^2 - x \mid x \in \mathbb{X} \rangle}, \quad l \leftarrow 1, \quad \varepsilon_1' \leftarrow \max\{1/2, \varepsilon'\}$
				
				\IF{$\mathcal{F}_1(\bm{1})=\bm{0}$}
				\RETURN $\bm{1}$
				\ENDIF
				
				\WHILE{$l \le \lceil \log_{\varepsilon_1'} \varepsilon' \rceil$}
				\STATE $\mathbb{Y} \leftarrow \mathbb{X}, \quad \mathcal{F}_1 \leftarrow \mathcal{F}_B$
				
				\WHILE{$\mathcal{F}_1 \neq \emptyset \land \mathcal{F}_1(\mathbf{0}) \neq \mathbf{0}$}
				\IF{$\mathcal{F}_1 \cap \mathbb{C} \setminus \{0\} \neq \emptyset$}
				\STATE $l \leftarrow l + 1, \quad \textbf{break}$ 
				\ENDIF
				
				\STATE $d \leftarrow \#\mathbb{Y}, \quad (\mathcal{B}_{\mathcal{F}_1, d}, \bm{b}_{\mathcal{F}_1, d}) \leftarrow \text{ConstructSys}(\mathcal{F}_1, d)$
				\STATE $\mathcal{F}_R \leftarrow \text{Red2}(\mathcal{F}_1), \quad (\mathcal{B}_R, \bm{b}_R) \xleftarrow{ \mathcal{F}_R } \text{ReConstructSys}(\mathcal{F}_1, d), \quad |\widehat{\bm{m}}_d\rangle \leftarrow \textbf{HHL}(\mathcal{B}_{R}, \bm{b}_{R}, \sqrt{\varepsilon'_1/n})$  
				
				\STATE $|e\rangle \xleftarrow{\text{Measure}} |\widehat{\bm{m}}_d\rangle, \quad m_{\bar{d},k} = \prod_{i=1}^{u_k} x_{n_i} \leftarrow \text{Identify}(|e\rangle)$
				\STATE $x_{n_i} \leftarrow 1 \; (\forall i \in \{1, \dots, u_k\}), \quad \mathcal{F}_1 \leftarrow \mathcal{F}_1\big|_{x_{n_i}=1} \setminus \{0\}, \quad \mathbb{Y} \leftarrow \mathbb{Y} \setminus \{x_{n_i} \mid i=1,\dots,u_k\}$
				\ENDWHILE
				
				\FOR{$i = 1$ \TO $n$}
				\IF{$x_i \in \mathbb{Y}$}
				\STATE $a_i \leftarrow 0$
				\ELSE
				\STATE $a_i \leftarrow 1$
				\ENDIF
				\ENDFOR
				\RETURN $\mathbf{a} = (a_1, a_2, \dots, a_n)$
				\ENDWHILE
				
				\RETURN $\emptyset$
			\end{algorithmic}
		\end{algorithm}
	\end{minipage} 
\end{figure*}
Analogous to the reduction process for standard Macaulay matrices detailed previously, we directly generate the Boolean Macaulay matrix from $\mathcal{F}_1$ and apply the Red2 to obtain the new system $\mathcal{F}_R$. We subsequently utilize $\mathcal{F}_R$ to substitute the corresponding $r$ rows in the original matrix, yielding a new Boolean Macaulay matrix $\mathcal{B}_{R}$. Because this procedure strictly equates to applying identical elementary row operations to both the matrix and the right-hand side vector, the solution space of the Boolean Macaulay linear system remains completely invariant. Combining the results of Ding et al.~\cite{ding2023limitations} with the reduction process, the sparsity of the reduced matrix $\mathcal{B}_{R}$ remains bounded by $\mathcal{O}(r\cdot T_{\mathcal{F}})$.
\par
\begin{theorem}When the equation system $\mathcal{F}=0$ admits a Boolean solution, Algorithm 5 outputs a Boolean solution with a probability of at least $1-\varepsilon'$. Its runtime complexity is bounded by $\mathcal{O}(n^{\frac{5}{2}}rT_{\mathcal{F}}\kappa^2 \log_{2} \frac{1}{\varepsilon'} )$, where $\kappa=\max \kappa_{\bm{b}_{R}}(\mathcal{B}_{R})$ denotes the maximum condition number of all Boolean Macaulay matrices $\mathcal{B}_{R}$ utilized in Step 15 with respect to the vector $\bm{b}_{R}$.
\end{theorem}
\begin{proof}
	Algorithm 5 perfectly mirrors the structure of Algorithm 3, exclusively substituting the Boolean Macaulay matrix and omitting field polynomials. Step 15 dominates the overall time complexity. It utilizes the HHL algorithm to solve the Boolean Macaulay linear system $\mathcal{B}_R\bm{m}_D=\bm{b}_R$, where the matrix dimensions reach at most $r2^n\times(2^n-1)$. The vector $\bm{b}_R$ strictly satisfies the conditions in Theorem V.2 corresponding to $\rho=r$, $m=r$, and $v=n$, bounding the state preparation time by $\mathcal{O}(\log_2((r+1)2^n-1))$. According to Chen and Gao~\cite{chen2022quantum}, a single execution of this step requires at most $c\log_2((r+1)2^n-1)rT_{\mathcal{F}}\kappa^2\sqrt{n/\varepsilon_1'}$ time. In the worst-case scenario where variable assignments fail to reduce matrix sparsity, the inner loop executes at most $n-1$ times, while the outer loop executes at most $\lceil \log_{\varepsilon_1'} \varepsilon' \rceil$ times. Setting $\varepsilon'_1 = 1/2$ and applying logarithmic scaling bounds the total time complexity strictly by $\mathcal{O}(n^{\frac{5}{2}}rT_{\mathcal{F}}\kappa^2 \log_{2} \frac{1}{\varepsilon'} )$.
\end{proof}
\par

By substituting Algorithm 5 for Algorithm 3 within Algorithm 4, we obtain Algorithm 6, an efficient quantum algorithm for solving Boolean polynomial systems via Boolean Macaulay matrices (Q-SBP-BM), and prove its properties. Fig.~\ref{fig3} delineates the complete flowchart of this algorithm.
\begin{figure*}[!t]
	\centering
	\includegraphics[width=6.3in]{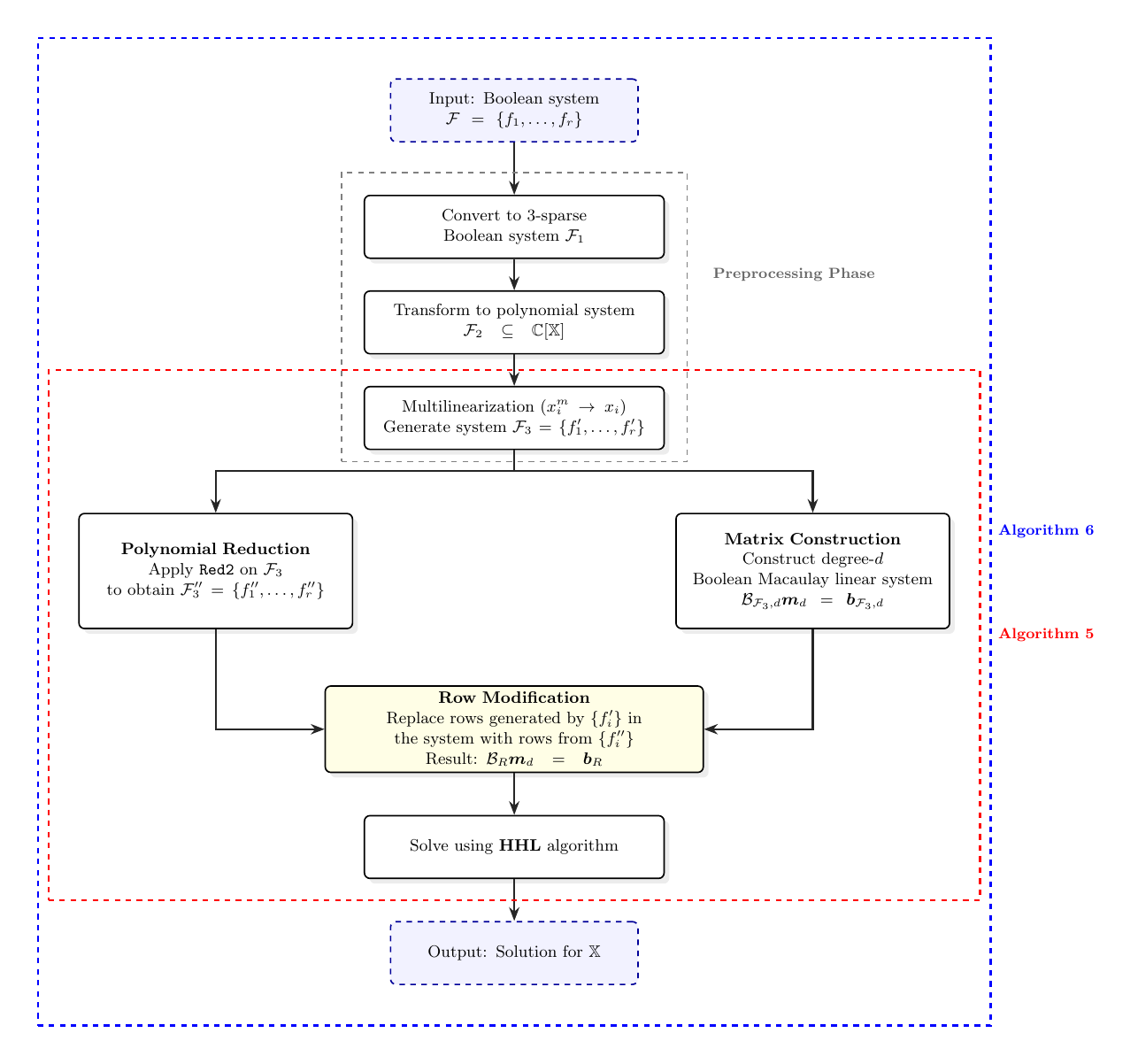}
	\caption{Boolean Macaulay-based quantum algorithm for Boolean systems.}
	\label{fig3}
\end{figure*}
\par
\begin{theorem}When the Boolean system $\mathcal{F}=0$ admits a solution, Algorithm 6 outputs a solution with a probability of at least $1-\varepsilon'$. Its runtime complexity is bounded by $\mathcal{O}((n^{\frac{5}{2}}+T_{\mathcal{F}}^{\frac{5}{2}})rT_{\mathcal{F}}\kappa^2\log_2 \frac{1}{\varepsilon'})$, where $\kappa=\max \kappa_{\bm{b}'_{R}}(\mathcal{B}'_{R})$ denotes the maximum condition number of all Boolean Macaulay matrices $\mathcal{B}'_{R}$ utilized with respect to specific vectors $\bm{b}'_{R}$ during the process of finding Boolean solutions to the system $\mathcal{F}_2=0$ over the complex field via Algorithm 5.
\end{theorem}
\begin{proof}
	Transforming the Boolean polynomials into a 3-sparse partition set in Step 2 and converting them into complex-field polynomials in Step 3 both require $\mathcal{O}(T_{\mathcal{F}})$ time. Since $\mathbb{V}_{\mathbb{F}_2}(\mathcal{F})=\operatorname{Proj}_{\mathbb{X}}\mathbb{V}_B(\widehat{C}(\mathcal{F}_1))$, the problem reduces exclusively to finding the Boolean solutions of $\mathcal{F}_2$. Employing an $s$-sparse partition bounds the variable count by $n+T_{\mathcal{F}}/s$ and the term count by $T_{\mathcal{F}}s^{\frac{s}{2}}$. Following Theorem V.5, Step 4 requires $\mathcal{O}((n+\frac{T_{\mathcal{F}}}{s})^{\frac{5}{2}}rT_{\mathcal{F}}s^{\frac{s}{2}}\kappa^2\log_2\frac{1}{\varepsilon'})$ time. In practice, setting $s=3$ strictly bounds the overall time complexity of Algorithm 6 by $\mathcal{O}((n^{\frac{5}{2}}+T_{\mathcal{F}}^{\frac{5}{2}})rT_{\mathcal{F}}\kappa^2\log_2\frac{1}{\varepsilon'})$.
\end{proof}

\begin{algorithm}[H]
	\caption{Q-SBP-BM ($\mathcal{F}, \varepsilon'$)}
	\label{alg:Q-SBP-BM}
	\begin{algorithmic}[1] 
		\REQUIRE $\mathcal{F} = \{f_1, \dots, f_r\} \subset \mathcal{R}_2[\mathbb{X}]$, $\varepsilon' \in (0, 1)$.
		\ENSURE $\bm{a} \in \mathbb{V}_{\mathbb{F}_2}(\mathcal{F}) \lor \emptyset$.
		
		\STATE $\mathbb{Y} \leftarrow \mathbb{X} \cup \mathbb{U}(\mathcal{F}, 3)$  
		\STATE $\mathcal{F}_1\leftarrow S(\mathcal{F}, 3) \subseteq \mathcal{R}_2[\mathbb{Y}]$  
		\STATE $\mathcal{F}_2 \leftarrow \widehat{C}(\mathcal{F}_1) \subseteq \mathbb{C}[\mathbb{Y}]$  
		\STATE $\bm{a} \leftarrow \text{Q-BSCP-BM} (\mathcal{F}_2, \varepsilon')$ 
		
		\IF{$\bm{a} \neq \emptyset$}
		\RETURN $\text{Proj}_{\mathbb{X}}\ \bm{a}$
		\ENDIF
		
		\RETURN $\emptyset$
	\end{algorithmic}
\end{algorithm}
\subsubsection{Bound Analysis of the Condition Number}
\label{section5-4-2}
Assuming variable assignments remain unknown, applying the Red2 transforms the system into $\mathcal{B}_R\bm{m}_d=\bm{b}_R$. Following the analogous analysis in Sections \ref{section2-3-1} and \ref{section5-3-2}, we have $\mathcal{B}_R^+\bm{b}_R=\mathcal{B}^+\bm{b}$ and $\|\bm{b}_R \|=(T_{\mathcal{F}}-2r)^{\frac{1}{2}}$. Retaining the assumptions that $\|\mathcal{B}_R\|=\|\mathcal{B}\|$ and $\|\mathcal{B}\|\geq 1$, we deduce $\kappa_{\bm{b}_R}(\mathcal{B}_R) \geq (T_{\mathcal{F}}-2r)^{-\frac{1}{2}}\|\mathcal{B}^+\bm{b}\|$. When the Boolean system admits a unique solution $\bm{a}$ with Hamming weight $h$, the condition number lower bound is $\kappa_{\bm{b}_R}(\mathcal{B}_R)\geq (T_{\mathcal{F}}-2r)^{-\frac{1}{2}}(2^h-1)^{\frac{1}{2}}$.
\par
Consider a 3-sparse Boolean polynomial system comprising 4 variables and 512 polynomials. Assume all polynomials contain constant terms. Each polynomial generates 16 rows, forming an $8192\times 15$ Boolean Macaulay matrix. By modifying the search algorithm based on structural constraints, $10^6$ iterations yielded a matrix instance with a minimized spectral norm bounded by $\|\mathcal{B}\|\approx 21.03$, structurally depicted in Fig.~\ref{fig4}.
\begin{figure*}[!t]
	\centering
	\subfloat[]{\includegraphics[width=3.4in]{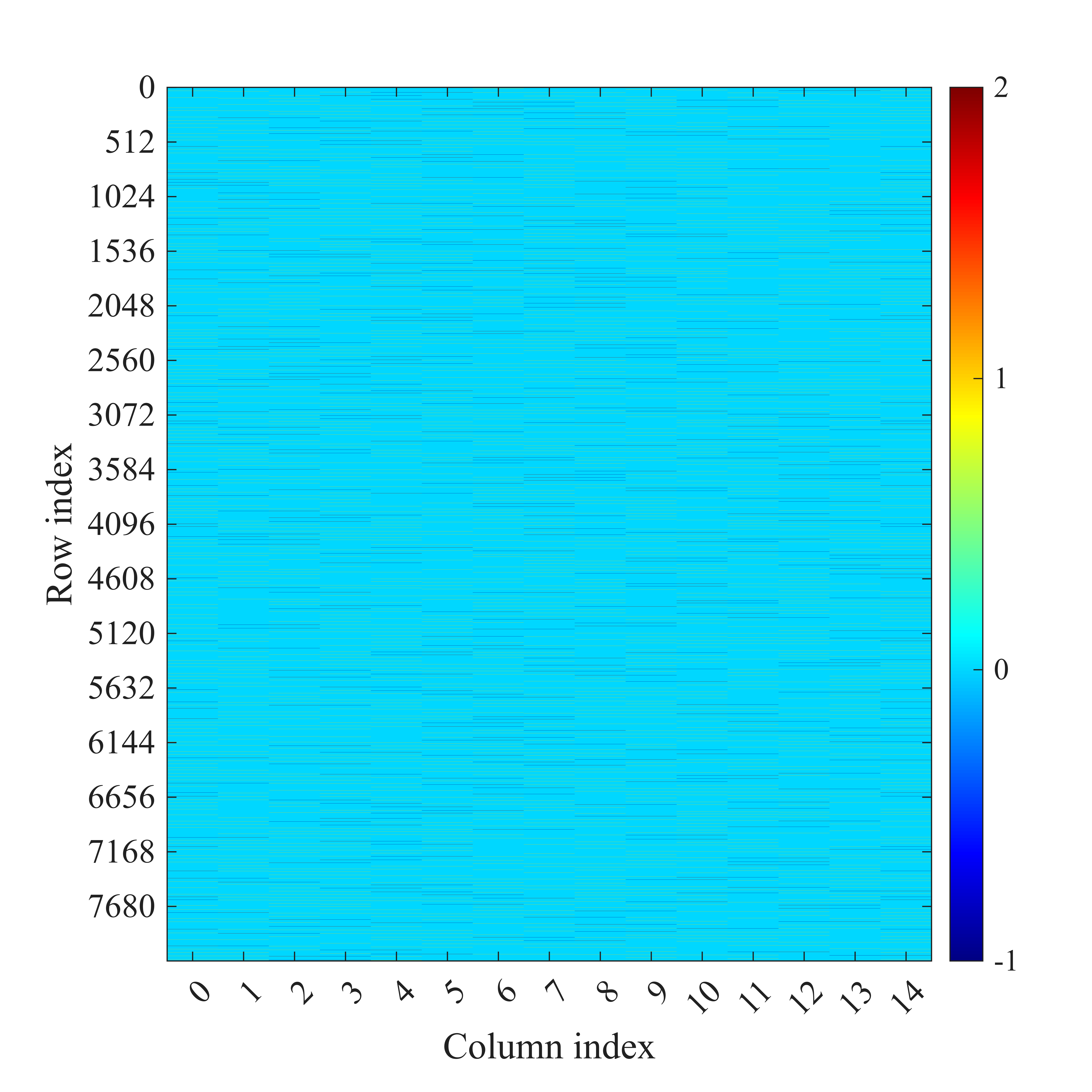} 
		\label{fig:sub3}}
	\hfil
	\subfloat[]{\includegraphics[width=3.4in]{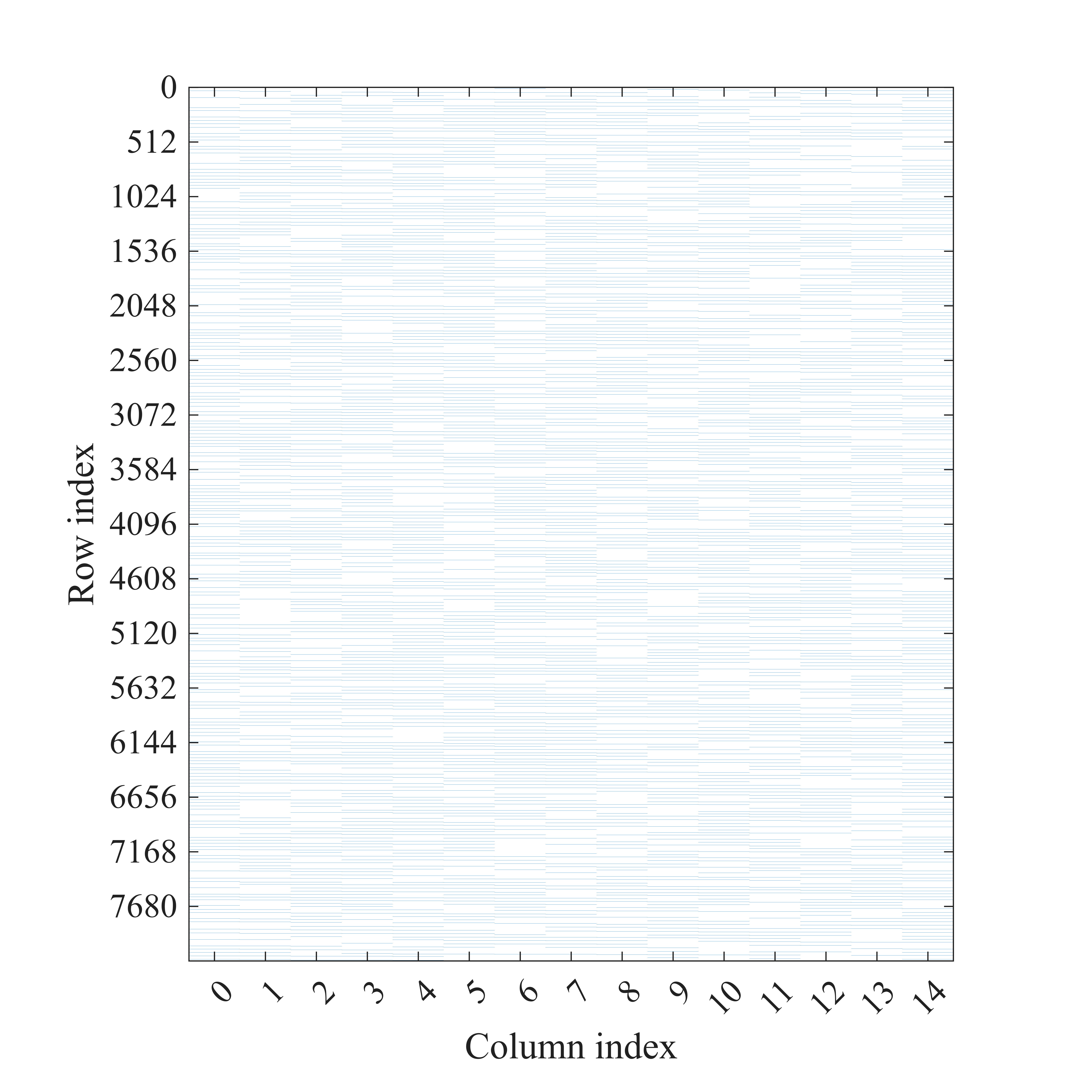}
		\label{fig:sub4}}
	\caption{Visualization of the Boolean Macaulay matrix with dimensions $8192 \times 15$. (a) The distribution of element values. (b) The sparsity pattern representing non-zero entries.}
	\label{fig4}
\end{figure*}
\par
The system corresponding to the aforementioned matrix $\mathcal{B}$ satisfies $T_{\mathcal{F}}=1536$ and $\|\bm{b}_{R}\|=\sqrt{512}$. Similarly, it can be proven that $T_{\mathcal{F}}^{-\frac{1}{2}}(2^h-1)^{\frac{1}{2}} < \kappa_{\bm{b}_R}(\mathcal{B}_R) < (2^h-1)^{\frac{1}{2}}$, conclusively demonstrating that our reduction can strictly surpass the theoretical condition number lower bound established by Ding et al.
\par
Incorporating the impact of row operations on the matrix norm, the condition number interval under the Red1 is approximately $[ (T_{\mathcal{F}}-2r)^{-\frac{1}{2}}(2^h-1)^{\frac{1}{2}}\|\mathcal{M}\|, (T_{\mathcal{F}}-2r)^{\frac{1}{2}}(2^h-1)^{\frac{1}{2}}\|\mathcal{M}\| ]$; whereas the Red2 strictly compresses this interval to $[ (T_{\mathcal{F}}-2r)^{-1}(2^h-1)^{\frac{1}{2}}\|\mathcal{M}\|, (2^h-1)^{\frac{1}{2}}\|\mathcal{M}\|]$. The significant optimization of both the upper and lower bounds further validates the superiority of our reduction method.
\subsubsection{Quantum Time Complexity Bounds for\\ Solving the LPSN Problem}
\par
\label{section5-4-3}
According to Theorem V.6, the time complexity of employing Algorithm 6 to solve the system $\mathcal{F}$ is bounded by $\mathcal{O}((n^{\frac{5}{2}} + T_{\mathcal{F}}^{\frac{5}{2}})rT_{\mathcal{F}}\kappa^2 \log_2 \frac{1}{\varepsilon'})$.
When possessing $\lceil 2^m\ln\frac{1}{\varepsilon}\rceil$ samples, to achieve a success probability of at least $(1-\varepsilon)(1-\varepsilon')$, the runtime of Algorithm 6 is bounded by $\mathcal{O}\left(2^{\frac{9}{2}m}n^{\frac{7}{2}d}\kappa^2\ln^{\frac{7}{2}}\frac{1}{\varepsilon}\log_2\frac{1}{\varepsilon'}\right)$.
\par
Before any variable assignments, initially solving the Boolean Macaulay linear system $\mathcal{B}_R\bm{m}_d=\bm{b}_R$ dictates $\kappa_{\bm{b}_R}(\mathcal{B}_R)^2 \geq T_{\mathcal{F}}^{-1}(2^h-1)$. Since $T_{\mathcal{F}}-2r=\mathcal{O} (2^mn^d\ln\frac{1}{\varepsilon} )$, we deduce $\kappa_{\bm{b}_R}(\mathcal{B}_R)^2 \geq 2^{h-m}n^{-d}\ln^{-1}\frac{1}{\varepsilon}$. Therefore, utilizing Algorithm 6 recovers the secret vector in the LPSN problem with a probability of at least $(1-\varepsilon)(1-\varepsilon')$ in a time complexity bounded by $\mathcal{O}\left(\Omega\left(2^{\frac{7}{2}m+h}n^{\frac{5}{2}d}\ln^{\frac{5}{2}}\frac{1}{\varepsilon}\log_2\frac{1}{\varepsilon'}\right)\right)$.


\section{Quantum Resource Evaluation}
\label{section6}
\subsection{Quantum Resource Estimate}
\label{section6-1}
Scherer \textit{et al.}~\cite{scherer2017concrete} provided the first concrete logical resource estimate for the HHL algorithm under the standard quantum-circuit model. Their conclusion can be summarized as follows: even in the idealized case that counts only the logical level, the HHL algorithm requires several hundred logical qubits and a gate count and circuit depth exceeding the order of \(10^{25}\); these resources are dominated by the Suzuki--Trotter time-splitting of Hamiltonian simulation, depend critically on the condition number \(\kappa\) and the precision \(\varepsilon\), and depend only logarithmically on the matrix size \(N\). In other words, what sets the order of magnitude of the HHL algorithm's resources is not the problem size itself but \(\kappa\) and \(\varepsilon\). The core step of our Algorithm~3 and Algorithm~5 invokes the HHL algorithm to solve a Macaulay and a Boolean-Macaulay linear system, and their dominant resources are inherited entirely from this primitive, on which basis we evaluate the required qubit count and circuit size in a fine-grained manner. To make this evaluation concrete, Fig.~\ref{fig:hhl-core} depicts the register-level circuit of one HHL-core invocation inside Algorithm~3 (\(A=M_R\)) or Algorithm~5 (\(A=B_R\)): the number of horizontal wires directly exhibits the circuit width, while the macro box \(e^{\mathrm{i}At}\) encapsulates the two nested multipliers that dominate the circuit depth, which are expanded in Figs.~\ref{fig:qpe-cascade} and~\ref{fig:trotter} below.

\begin{figure*}[!t]
	\centering
	\resizebox{0.95\textwidth}{!}{%
		\begin{quantikz}[column sep=0.45cm, row sep=0.6cm]
			\lstick{\makecell[r]{index reg.\\ $|0\rangle^{\otimes n_{\mathrm{idx}}}$}}
			& \gate{\mathrm{Prep}\,|b_R\rangle} & \qw & \gate{e^{\mathrm{i}At}} & \qw & \qw & \qw & \qw & \qw & \gate{e^{-\mathrm{i}At}} & \qw & \meter{} \\
			\lstick{\makecell[r]{phase reg.\\ $|0\rangle^{\otimes n_{\mathrm{ev}}}$}}
			& \qw & \gate{H^{\otimes n_{\mathrm{ev}}}} & \ctrl{-1} & \gate{\mathrm{QFT}^{\dagger}} & \ctrl{1} & \qw & \ctrl{1} & \gate{\mathrm{QFT}} & \ctrl{-1} & \gate{H^{\otimes n_{\mathrm{ev}}}} & \qw \\
			\lstick{\makecell[r]{eigen-inverse reg.\\ $|0\rangle^{\otimes n_{\mathrm{ev}}}$}}
			& \qw & \qw & \qw & \qw & \gate{1/\lambda} & \ctrl{1} & \gate{(1/\lambda)^{\dagger}} & \qw & \qw & \qw & \qw \\
			\lstick{\makecell[r]{rotation anc.\\ $|0\rangle$}}
			& \qw & \qw & \qw & \qw & \qw & \gate{R_y\!\left(2\arcsin\tfrac{C}{\lambda}\right)} & \qw & \qw & \qw & \qw & \meter{}
		\end{quantikz}%
	}
	\caption{Register-level quantum circuit of one HHL-core invocation inside Algorithm~3 (\(A=M_R\)) or Algorithm~5 (\(A=B_R\)). The Red2 reduction and the construction of the reduced matrix \(A\in\{M_R,B_R\}\) and right-hand side \(b_R\) are performed classically and are not depicted. Between successive invocations a monomial is identified, variables are assigned classically, and a smaller system is reconstructed and re-solved. The two measurements (index register and rotation ancilla) occur at the same final time step, after the eigenvalue computation has been fully uncomputed. The total number of wires realizes the circuit width \(W = n_{\mathrm{idx}} + 2\,n_{\mathrm{ev}} + \mathcal{O}(b_{\mathrm{fix}}) + \mathcal{O}(1)\) analyzed in the text, with \(n_{\mathrm{idx}}=n\lceil\log_2(3n+1)\rceil\) for Algorithm~3 and \(n_{\mathrm{idx}}=n\) for Algorithm~5.}
	\label{fig:hhl-core}
\end{figure*}
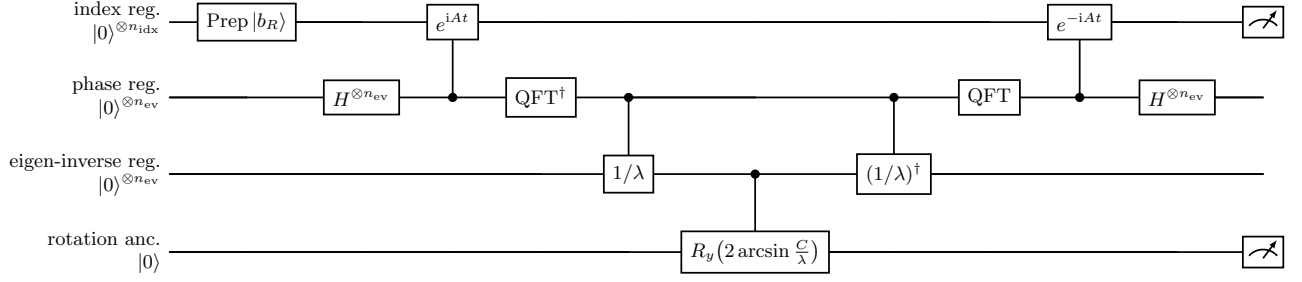

The qubit count, namely the circuit width, is obtained by summing the sizes of the registers used by the HHL algorithm, i.e., the wires of Fig.~\ref{fig:hhl-core}. To solve a linear system, the HHL algorithm requires an index register encoding the solution vector, whose size equals the binary logarithm of the number of matrix columns; a phase register storing the eigenvalues, whose size is jointly determined by precision and condition number and is taken as \(n_{\mathrm{ev}}=\lceil\log_2(\kappa/\varepsilon')\rceil+\mathcal{O}(1)\); an ancilla register of the same size storing the reciprocals of the eigenvalues; a single-qubit ancilla for the controlled rotation; and several fixed-point arithmetic ancilla registers of bit width \(b_{\mathrm{fix}}\). The circuit width can thus be written as
\[
W \;=\; n_{\mathrm{idx}} \;+\; 2\lceil\log_2(\kappa/\varepsilon')\rceil \;+\; \mathcal{O}(b_{\mathrm{fix}}) \;+\; \mathcal{O}(1),
\]
where \(n_{\mathrm{idx}}\) is the index-register size and constitutes the only essential difference between the two matrix types, as annotated in the caption of Fig.~\ref{fig:hhl-core}.

For the Macaulay linear system solved by Algorithm~4, the qubit count is governed mainly by the number of columns of the Macaulay matrix. This matrix \(M_R\in\mathbb{C}^{(r+n)2^{n\delta}\times(2^{n\Delta}-1)}\) has \(2^{n\Delta}-1\) columns, so the index register requires \(n_{\mathrm{idx}}=n\Delta\) qubits, where \(\Delta=\lceil\log_2(3n+1)\rceil\) corresponds to the full solving degree \(D=3n\). Substituting into the expression above yields the circuit width
\[
W_{\text{Mac}} \;\approx\; n\,\lceil\log_2(3n+1)\rceil \;+\; 2\lceil\log_2(\kappa/\varepsilon')\rceil \;+\; \mathcal{O}(b_{\mathrm{fix}}),
\]
which grows quasi-linearly in the number of unknowns \(n\).

For the Boolean-Macaulay linear system solved by Algorithm~6, the qubit count is reduced further. This matrix \(B_R\) has dimensions at most \(r2^n\times(2^n-1)\) with \(2^n-1\) columns, so the index register requires only \(n_{\mathrm{idx}}=n\) qubits. Substituting into the expression above yields the circuit width
\[
W_{\text{BMac}} \;\approx\; n \;+\; 2\lceil\log_2(\kappa/\varepsilon')\rceil \;+\; \mathcal{O}(b_{\mathrm{fix}}),
\]
which grows only linearly in the number of unknowns \(n\), saving a factor of about \(\lceil\log_2(3n+1)\rceil\) in index overhead compared with the Macaulay case. When \(\kappa\) and \(\varepsilon'\) take common magnitudes, this width is on the order of a few hundred, consistent with the baseline of~\cite{scherer2017concrete}.

The circuit depth and total gate count are determined by the nested loops within the HHL algorithm, whose dominant term is the Trotter time-splitting of Hamiltonian simulation. Specifically, the total gate count and depth are obtained by multiplying the layers in turn:
\[
D \;\sim\; N_{\mathrm{HHL}}\,\times\, N_{\mathrm{QPE}}\,\times\, r\,m_s\,\times\, g_{\mathrm{slice}},
\]
where \(N_{\mathrm{HHL}}\) is the number of HHL-algorithm calls, \(N_{\mathrm{QPE}}\) the number of phase-estimation iterations, \(m_s\) the number of \(1\)-sparse submatrices, and \(g_{\mathrm{slice}}\) the gate count per Trotter slice, which includes the gates needed to compute the matrix entries. The first nested multiplier, \(N_{\mathrm{QPE}}\), is visualized in Fig.~\ref{fig:qpe-cascade}, which expands the macro box \(e^{\mathrm{i}At}\) of Fig.~\ref{fig:hhl-core} into the complete QPE controlled-power cascade: the \(n_{\mathrm{ev}}\) controlled powers \(e^{\mathrm{i}2^{j}At}\) accumulate a total evolution of \(N_{\mathrm{QPE}}\sim 2^{n_{\mathrm{ev}}}\sim \mathcal{O}(\kappa/\varepsilon')\) elementary time steps, in exact agreement with the \(\kappa\)-dependence entering through the Hamiltonian-simulation time below.

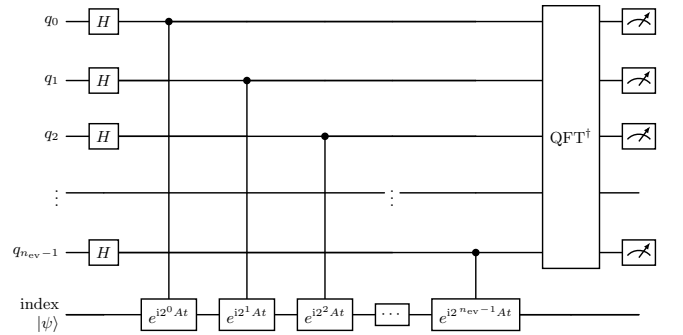
\begin{figure}[!t]
	\centering
	\resizebox{\columnwidth}{!}{%
		\begin{quantikz}[column sep=0.45cm, row sep=0.6cm]
			\lstick{$q_0$}
			& \gate{H} & \ctrl{5} & \qw & \qw & \qw & \qw & \gate[5]{\mathrm{QFT}^{\dagger}} & \meter{} \\
			\lstick{$q_1$}
			& \gate{H} & \qw & \ctrl{4} & \qw & \qw & \qw & & \meter{} \\
			\lstick{$q_2$}
			& \gate{H} & \qw & \qw & \ctrl{3} & \qw & \qw & & \meter{} \\
			\lstick{$\vdots$}
			& & & & & \ \vdots\  & & & \\
			\lstick{$q_{n_{\mathrm{ev}}-1}$}
			& \gate{H} & \qw & \qw & \qw & \qw & \ctrl{1} & & \meter{} \\
			\lstick{\makecell[r]{index\\ $|\psi\rangle$}}
			& \qw
			& \gate{e^{\mathrm{i}2^{0}At}}
			& \gate{e^{\mathrm{i}2^{1}At}}
			& \gate{e^{\mathrm{i}2^{2}At}}
			& \gate{\cdots}
			& \gate{e^{\mathrm{i}2^{\,n_{\mathrm{ev}}-1}At}}
			& \qw & \qw
		\end{quantikz}%
	}
	\caption{Inside the box \(e^{\mathrm{i}At}\) of Fig.~\ref{fig:hhl-core}: the complete QPE controlled-power cascade, constituting depth source~(I). The geometrically increasing controlled powers contribute the multiplicative factor \(N_{\mathrm{QPE}}\sim 2^{n_{\mathrm{ev}}}\sim \mathcal{O}(\kappa/\varepsilon')\) to the circuit depth.}
	\label{fig:qpe-cascade}
\end{figure}

The second nested multiplier originates from the Suzuki--Trotter decomposition of each controlled power, shown in Fig.~\ref{fig:trotter}: every evolution \(e^{\mathrm{i}A\tau}\) is realized by \(m_s\) one-sparse walk steps \(U_{\mathrm{walk}}^{(A_j)}\), each sandwiched between an oracle call \(\mathcal{O}_A\) and its inverse, and \(r\) consecutive such slices reproduce \(e^{\mathrm{i}A\tau}\) to precision \(\varepsilon'\), giving exactly the factor \(r\,m_s\) in the depth product above. The time-splitting factor is
\[
r \;=\; 5^{\,k-1/2}\bigl(2m_s\|A\|t\bigr)^{1+1/2k}\big/{\varepsilon'}^{1/2k}, \qquad t\sim \mathcal{O}(\kappa/\varepsilon').
\]
Taking the Suzuki order \(k=2\) and the sparse decomposition \(s=3\), and noting the lack of parallelism in the HHL algorithm~\cite{scherer2017concrete}, the circuit depth and total gate count are of the same order; multiplying the layers of Figs.~\ref{fig:hhl-core}--\ref{fig:trotter} in turn, their magnitude coincides exactly with the time complexities given by Theorems~V.4 and~V.6, which validates our theoretical analysis at the circuit level.

\begin{figure}[!t]
	\centering
	\resizebox{\columnwidth}{!}{%
		\begin{quantikz}[column sep=0.45cm, row sep=0.6cm]
			\lstick{\makecell[r]{phase ctrl\\ $t[j]$}}
			& \qw
			& \ctrl{1}
			& \qw
			& \qw
			& \qw
			& \ctrl{1}
			& \qw
			& \qw \\
			\lstick{\makecell[r]{working reg.}}
			& \gate[2]{O_A}
			& \gate{U_{\mathrm{walk}}^{(A_1)}}
			& \gate[2]{O_A^{\dagger}}
			& \gate{\cdots}
			& \gate[2]{O_A}
			& \gate{U_{\mathrm{walk}}^{(A_{m_s})}}
			& \gate[2]{O_A^{\dagger}}
			& \qw \\
			\lstick{\makecell[r]{oracle anc.\\ $|0\rangle^{\otimes b_{\mathrm{fix}}}$}}
			& & \qw & & \qw & & \qw & & \qw
		\end{quantikz}%
	}
	\caption{Inside one \(e^{\mathrm{i}A\tau}\): the Suzuki--Trotter expansion, constituting depth source~(II). Two \(1\)-sparse steps (bands \(A_1,\dots,A_{m_s}\)) are shown explicitly; \(m_s\) such steps form one Trotter slice, and \(r\) consecutive slices realize \(e^{\mathrm{i}A\tau}\), giving the factor \(r\,m_s\) in the depth estimate.}
	\label{fig:trotter}
\end{figure}
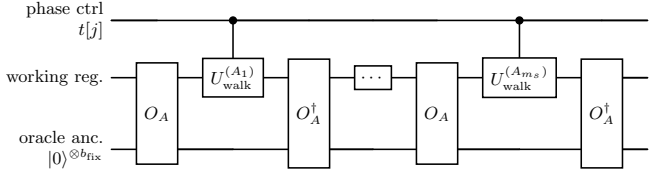

For the Macaulay linear system, the circuit depth and gate count are of the same order as in Theorem~V.4. Here the matrix sparsity is at most \(2T_F+2n\), and by~\cite{chen2022quantum} it admits an efficient \(1\)-sparse decomposition, so that
\[
D_{\text{Mac}} \;\sim\; \mathcal{O} \Bigl((n^{7/2}+T_F^{7/2})\,\kappa^2\,\log_2\tfrac{1}{\varepsilon'}\Bigr),
\]
where the condition number satisfies \(\kappa\ge(T_F-2r)^{-1/2}\bigl(\binom{3n+h}{h}-1\bigr)^{1/2}\).

For the Boolean-Macaulay linear system, the circuit depth and gate count are of the same order as in Theorem~V.6. Here the matrix sparsity is \(\mathcal{O}(rT_F)\), so that
\[
D_{\text{BMac}} \;\sim\; \mathcal{O}\Bigl((n^{5/2}+T_F^{5/2})\,rT_F\,\kappa^2\,\log_2\tfrac{1}{\varepsilon'}\Bigr),
\]
where the condition number satisfies \(\kappa\ge(T_F-2r)^{-1/2}(2^{h}-1)^{1/2}\).

The above results show that both the circuit depth and the gate count grow as the square of the condition number, so the optimization of the condition number by Red2 translates directly into a resource saving. On the one hand, as is evident from Figs.~\ref{fig:qpe-cascade} and~\ref{fig:trotter}, the condition number enters the circuit depth through the Hamiltonian-simulation time \(t\sim\kappa/\varepsilon'\) and the time-splitting factor \(r\sim\kappa^{1+1/2k}\); on the other hand, it appears as \(\kappa^2\) in the total complexity. Relative to the Red1 employed by Ding \textit{et al.}~\cite{ding2023limitations}, our Red2 tightens the upper and lower bounds of the condition number with respect to the right-hand-side vector without altering the solution space of the linear system, thereby lowering the circuit depth and gate count by a corresponding amount. This is precisely the resource-level manifestation of our condition-number optimization.

Combining the two aspects of qubits and circuit size, we can compare the resources of the two quantum algorithms. In terms of qubits, the index register of the Boolean-Macaulay system is only \(n\), markedly smaller than the \(n\lceil\log_2(3n+1)\rceil\) of the Macaulay system; in terms of the condition number, the former depends on \(2^{h}\) and the latter on \(\binom{3n+h}{h}\), so when the Hamming weight \(h\) of the solution is small, the Boolean-Macaulay system has a lower condition-number bound and a smaller circuit size. Thus, with respect to resource consumption, the Boolean-Macaulay system is generally preferable, whereas in the specific parameter regimes identified in Section \ref{section7} the Macaulay system may still prevail; the two can be selected according to the algorithm-selection strategy.

\subsection{Feasibility and Quantum Advantage}
\label{section6-2}
The above resource estimate shows that our quantum algorithm is implementable at the logical level. The circuit width grows linearly or quasi-linearly in the number of unknowns \(n\), and only logarithmically in the condition number and precision; the circuit depth and gate count are of the same order as the time complexities given by Theorems~V.4 and~V.6, constituting finite and explicitly quantifiable closed-form expressions whose magnitude is governed chiefly by the condition number \(\kappa\). The circuit framework of Figs.~\ref{fig:hhl-core}--\ref{fig:trotter} makes this claim fully constructive: every layer of the depth product corresponds to an explicitly drawn sub-circuit. Hence the resource requirements of our quantum algorithm are definite and computable, not merely asymptotically meaningful.

The above estimate is conservative and admits further room for compression. The Hamiltonian-simulation calibration adopted here follows the early technique used in~\cite{scherer2017concrete}; employing the advanced Hamiltonian-simulation methods we already cite, namely those of Berry \textit{et al.}~\cite{berry2015hamiltonian} and Childs \textit{et al.}~\cite{childs2017quantum}, which improve the precision dependence to \(\mathcal{O}(\log(1/\varepsilon))\) and reduce the sparsity dependence to near-linear, would substantially further lower the circuit depth and gate count. Moreover, by lowering the error requirement to \(\sqrt{\varepsilon_1'/n}\) and iterating, our strategy compresses the \(1/\varepsilon'\) dependence to logarithmic, superior to the \(1/\varepsilon'\) of the naive HHL algorithm, rendering the precision-direction cost more controllable.

Taken together, our resource evaluation supports the potential of the quantum algorithm to outperform classical algorithms. Since the circuit size grows as the square of the condition number, and our Red2 tightens the condition-number interval relative to Ding \textit{et al.}~\cite{ding2023limitations}, a corresponding reduction in resources follows; combined with the complexity comparison in Section \ref{section7}, even at \(h=\Theta(\sqrt{n})\) our quantum algorithm retains the potential to outperform classical algorithms in specific parameter regimes. It follows that, with sparsity \(s=3\) and a fixed number of unknowns \(n\), the required qubit count is linear or quasi-linear in \(n\) and the circuit size is controllable, so that a controllable implementation cost and a potential speedup hold simultaneously.

\section{Comparative Analysis of Classical and Quantum Algorithms}
\label{section7}
Table~\ref{tab1} summarizes the requirements regarding the correlation polynomial, sample complexity, and time complexity of the four algorithms for solving the LPSN problem under both random and adversarial structured noise patterns. The success probabilities of the classical and quantum algorithms are at least $1-\varepsilon$ and $(1-\varepsilon)(1-\varepsilon')$, respectively.

\begin{table*}[htbp]
	\begin{center}
		\caption{Comparison of algorithms under random and adversarial structured noise patterns.}
		\label{tab1}
		\scalebox{0.9}{
			\begin{tabular}{|c|c|c|c|c|}
				\hline
				\multirow{2}{*}{Algorithm} & 
				\multirow{2}{*}{Noise pattern} &
				\multirow{2}{*}{Algebraic constraints} &
				\multirow{2}{*}{Sample complexity} &
				\multirow{2}{*}{Time complexity}  \\
				& & & & \\
				\hline
				
				\multirow{2}{*}{Classical algorithm 1~\cite{arora2011new}} 
				& Random & --- & $\Omega((\sum_{i=1}^d\binom{n}{i}+\log_2\frac{1}{\varepsilon})2^d)$ & $\mathcal{O}((n^d+\log_2\frac{1}{\varepsilon})2^dn^{2d+1})$ \\
				\cline{2-5}
				& Adversarial & --- & --- & --- \\
				\hline
				
				\multirow{2}{*}{Classical algorithm 2~\cite{arora2011new}} 
				& Random & ---  & $\Omega((\sum_{i=1}^d\binom{n}{i}+\log_2\frac{1}{\varepsilon})2^{m+d})$ & $\mathcal{O}((n^d+\log_2\frac{1}{\varepsilon})2^{m+d}n^{2d})$ \\
				\cline{2-5}
				& Adversarial & $\checkmark$ & $\Omega((\sum_{i=1}^{d'}\binom{n}{i}+\log_2\frac{1}{\varepsilon})2^{d'})$ & $\mathcal{O}((n^{d'}+\log_2\frac{1}{\varepsilon})2^{m+d'}n^{2d'})$ \\
				\hline
				
				\multirow{2}{*}{Quantum algorithm 4} 
				& Random & \multirow{2}{*}{$\checkmark$} & \multirow{2}{*}{$\Omega(2^m\ln\frac{1}{\varepsilon})$} & \multirow{2}{*}{$\mathcal{O} (\Omega (2^{\frac{5}{2}m}h^{-h}n^{\frac{5}{2}d+h}\ln^{\frac{5}{2}} \frac{1}{\varepsilon} \log_2\frac{1}{\varepsilon'} ) )$} \\
				\cline{2-2}
				& Adversarial & & & \\
				\hline
				
				\multirow{2}{*}{Quantum algorithm 6} 
				& Random & \multirow{2}{*}{$\checkmark$} & \multirow{2}{*}{$\Omega(2^m\ln\frac{1}{\varepsilon})$}  & \multirow{2}{*}{$\mathcal{O}(\Omega(2^{\frac{7}{2}m+h}n^{\frac{5}{2}d}\ln^{\frac{5}{2}}\frac{1}{\varepsilon}\log_2\frac{1}{\varepsilon'}))$} \\
				\cline{2-2}
				& Adversarial & & & \\
				\hline
				
			\end{tabular}
		}
	\end{center}
\end{table*}
To facilitate the comparison of time complexities, we extract their dominant exponential growth terms. Let $n$ denote the secret vector dimension, $h$ its Hamming weight, $m$ the number of correlated noise variables, and $d$ the correlation polynomial degree. We define $T_{c_1}=2^dn^{3d+1}$ and $T_{c_2}=2^{m+d}n^{3d}$ as the dominant time complexity terms for Classical Algorithms 1 and 2, and $T_{q_1}=2^{\frac{5}{2}m}h^{-h}n^{\frac{5}{2}d+h}$ and $T_{q_2}=2^{\frac{7}{2}m+h}n^{\frac{5}{2}d}$ for Quantum Algorithms 4 and 6, respectively. Due to space constraints, this paper exclusively presents the experimental results comparing the time complexities of Classical Algorithm 2 and the quantum algorithms.
\subsection{Selection Strategies for Classical Algorithms}
\label{section7-1}
Concerning applicable noise patterns, Classical Algorithm 2 exhibits broader applicability, functioning under both random and adversarial structured noise patterns, whereas Classical Algorithm 1 is exclusively viable for the random structured noise pattern.
\par 
With respect to sample complexity, under the random structured noise pattern, the sample size required by Classical Algorithm 1 is merely a $2^{-m}$ fraction of that required by Classical Algorithm 2, demonstrating a substantial advantage.
\par 
In terms of runtime, omitting logarithmic factors, consider that the condition $\frac{T_{c_1}}{T_{c_2}} \geq 1$ is equivalent to $n \geq 2^m$. When parameters satisfy the following conditions, we can determine the performance superiority: (1) $n < 2^m$, Classical Algorithm 1 outperforms Classical Algorithm 2; (2) $n = 2^m$, the upper bounds of their time complexities are identical; (3) $n > 2^m$, Classical Algorithm 2 outperforms Classical Algorithm 1.
\par 
Consequently, we provide the following selection strategies. Under the random structured noise pattern, Classical Algorithm 1 is advantageous when the available sample size is relatively small (e.g., $S=\Theta((\sum_{i=1}^d\binom{n}{i}+\log_2\frac{1}{\varepsilon})2^d)$) or parameters satisfy $n < 2^m$. Conversely, when an ample sample size is accessible (e.g., $S=\Omega((n^d+\log_2\frac{1}{\varepsilon})2^{m+d})$) and $n > 2^m$, Classical Algorithm 2 is preferable. Under the adversarial structured noise pattern, Classical Algorithm 2 remains the uniquely applicable classical approach, guaranteeing recovery given sufficient samples.

\subsection{Selection Strategies for Quantum Algorithms}
\label{section7-2}
Concerning applicable noise patterns, both Quantum Algorithm 4 and Quantum Algorithm 6 exhibit identical capabilities, functioning seamlessly under both random and adversarial structured noise patterns. To facilitate a concise comparison, we exclusively consider the adversarial structured noise pattern here, as the required complexities remain strictly identical across both patterns.
\par 
With respect to sample complexity, both algorithms merely require guaranteeing the uniqueness of the solution to the system of non-linear Boolean equations. Consequently, their sample requirements are strictly identical, both necessitating exactly $\lceil 2^m\ln\frac{1}{\varepsilon}\rceil$ queries to the oracle $Q_2(n,m,P,\mu,\bm{s})$ to achieve a success probability of at least $(1-\varepsilon)(1-\varepsilon')$.
\par 
In terms of runtime, omitting logarithmic factors, consider that the condition $\frac{T_{q_1}}{T_{q_2}} \geq 1$ is equivalent to $n \geq 2^{\frac{h\log_2h+m+h}{h}}$, which can also be elegantly expressed as $m \leq h(\log_2n-\log_2h-1)$. When parameters satisfy the following conditions, we can determine the performance superiority: (1) $m > h(\log_2n-\log_2h-1)$ (or $n < 2^{\frac{h\log_2h+m+h}{h}}$), Quantum Algorithm 4 outperforms Quantum Algorithm 6; (2) $m = h(\log_2n-\log_2h-1)$, the upper bounds of their time complexities are comparable; (3) $m < h(\log_2n-\log_2h-1)$ (or $n > 2^{\frac{h\log_2h+m+h}{h}}$), Quantum Algorithm 6 outperforms Quantum Algorithm 4.
\par 
Consequently, we provide the following selection strategies. Because the sample complexities are identical, the selection hinges entirely on the parameter regimes governing execution speed. Quantum Algorithm 4 is preferable when the sample noises are tightly correlated and the number of correlated noise variables is relatively large (e.g., $m=\Omega(\log_2n)$ or bounded by $m > \frac{n}{2e\ln 2}$). Conversely, when the relationship between sample noises is loose and the number of correlated noise variables is relatively small (e.g., $m=\mathcal{O}(\log_2n)$), Quantum Algorithm 6 is the superior choice.

\subsection{Comparative Analysis and Advantage Assessment of Classical versus\\ Quantum Algorithms}
\label{section7-3}
\subsubsection{Classical Algorithms versus the Quantum Algorithm\\ for Macaulay Linear Systems}
\label{section7-3-1}
Regarding sample complexity, the sample size required by Quantum Algorithm 4 is merely a $2^{m-d}n^{-d}$ fraction of that required by Classical Algorithm 1, demonstrating a substantial advantage.
\par
In terms of runtime, omitting logarithmic factors, consider that the condition $\frac{T_{c_1}}{T_{q_1}}\geq 1$ is equivalent to $2^{d-\frac{5}{2}m}h^h\geq n^{h-\frac{d}{2}-1}$. When parameters satisfy the following conditions, Quantum Algorithm 4 holds the potential to outperform Classical Algorithm 1: (1) $h<\frac{d}{2}+1$ and $n\geq 2^{\frac{2d-5m+2h\log_2h}{2h-d-2}}$, (2) $h=\frac{d}{2}+1$ and $(\frac{d}{2}+1)\log_2(\frac{d}{2}+1)+d\geq \frac{5}{2}m$, (3) $h>\frac{d}{2}+1$ and $1\leq n\leq 2^{\frac{2d-5m+2h\log_2h}{2h-d-2}}$.
\par
Considering that practically constructed Boolean equations containing at most $\mathcal{O}(n^t)$ terms and satisfying $1\leq t\leq d$, the dominant term of Quantum Algorithm 4 reduces to $T_{q_1}'=2^{\frac{5}{2}m}h^{-h}n^{\frac{5}{2}t+h}$, whereas Classical Algorithm 1 remains unaffected. Considering $\frac{T_{c_1}}{T_{q_1}'}\geq 1$, when parameters satisfy the following conditions, Quantum Algorithm 4 holds a potential advantage: (1) $\frac{5}{2}t+h<3d+1$ and $n\geq 2^{\frac{2d-5m+2h\log_2h}{5t+2h-6d-2}}$, (2) $\frac{5}{2}t+h=3d+1$ and $(3d+1-\frac{5}{2}t)\log_2(3d+1-\frac{5}{2}t)+d\geq \frac{5}{2}m$, (3) $\frac{5}{2}t+h>3d+1$ and $1\leq n\leq 2^{\frac{2d-5m+2h\log_2h}{5t+2h-6d-2}}$.
\par
Comparing Classical Algorithm 2 with Quantum Algorithm 4, the sample sizes required by Quantum Algorithm 4 under random and adversarial structured noise patterns are merely $(2n)^{-d}$ and $2^{m-d}n^{-d}$ fractions of those for Classical Algorithm 2, respectively. Regarding runtime, consider that the condition $\frac{T_{c_2}}{T_{q_1}}\geq 1$ is equivalent to $2^{d-\frac{3}{2}m}h^h\geq n^{h-\frac{d}{2}}$. When parameters satisfy the following conditions, Quantum Algorithm 4 holds the potential to outperform Classical Algorithm 2: (1) $h<\frac{d}{2}$ and $n\geq 2^{\frac{2d-3m+2h\log_2h}{2h-d}}$, (2) $h=\frac{d}{2}$ and $d\log_2(\frac{d}{2})+2d\geq 3m$, (3) $h>\frac{d}{2}$ and $1\leq n\leq 2^{\frac{2d-3m+2h\log_2h}{2h-d}}$. Fig.~\ref{fig5} validates this conclusion. Under the constraint of the actual term parameter $t$, considering $\frac{T_{c_2}}{T_{q_1}'}\geq 1$, when parameters satisfy the following conditions, Quantum Algorithm 4 holds a potential advantage: (1) $\frac{5}{2}t+h<3d$ and $n\geq 2^{\frac{2d-3m+2h\log_2h}{5t+2h-6d}}$, (2) $\frac{5}{2}t+h=3d$ and $(3d-\frac{5}{2}t)\log_2(3d-\frac{5}{2}t)+d\geq \frac{3}{2}m$, (3) $\frac{5}{2}t+h>3d$ and $1\leq n\leq 2^{\frac{2d-3m+2h\log_2h}{5t+2h-6d}}$.

\begin{figure*}[!t]
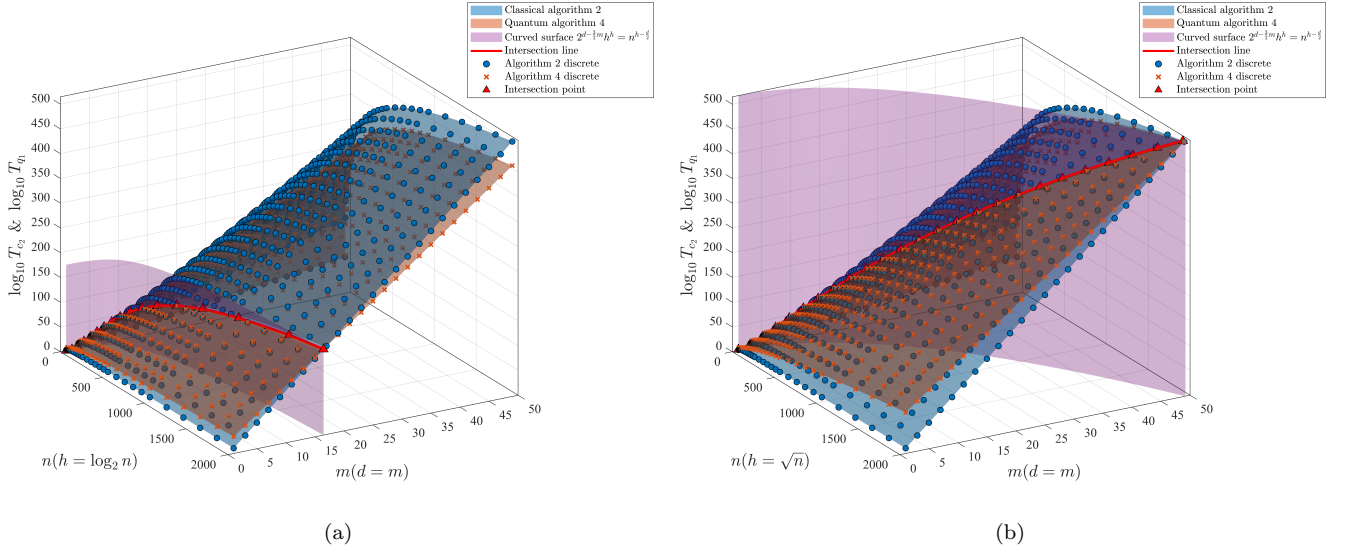

	\centering
	\subfloat[]{\includegraphics[width=3.4in]{2algorithm_comparison_combined_1200dpi.png} 
		\label{fig:sub5}}
	\hfil
	\subfloat[]{\includegraphics[width=3.4in]{2algorithm_comparison_combined_hsqrt_1200dpi.png}
		\label{fig:sub6}}
	\caption{Time complexity of algorithms 2 and 4 with $d=m$. (a) $h=\log_2n$. (b) $h=\sqrt n$.}
	\label{fig5}
\end{figure*}
\vspace{-1em}
\begin{figure*}[!t]
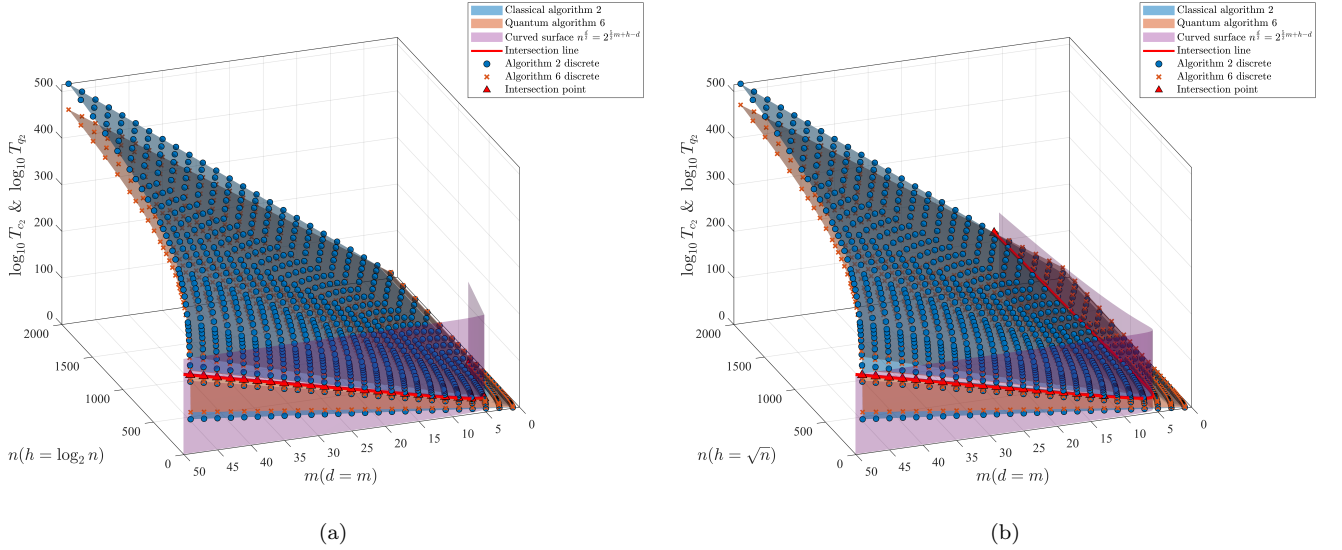

	\centering
	\subfloat[]{\includegraphics[width=3.4in]{4algorithm_comparison_combined_1200dpi.png} 
		\label{fig:sub7}}
	\hfil
	\subfloat[]{\includegraphics[width=3.4in]{4algorithm_comparison_combined_hsqrt_1200dpi.png}
		\label{fig:sub8}}
	\caption{Time complexity of algorithms 2 and 6 with $d=m$. (a) $h=\log_2n$. (b) $h=\sqrt n$.}
	\label{fig6}
\end{figure*}

\subsubsection{Classical Algorithms versus the Quantum Algorithm\\ for Boolean Macaulay Linear Systems}
\label{section7-3-2}
Because the sample requirements for Quantum Algorithm 6 and Quantum Algorithm 4 are strictly identical, Quantum Algorithm 6 maintains a substantial sample complexity advantage over the classical algorithms.
\par
Comparing Classical Algorithm 1 with Quantum Algorithm 6, consider that the condition $\frac{T_{c_1}}{T_{q_2}}\geq 1$ is equivalent to $n^{\frac{d}{2}+1}\geq 2^{\frac{7}{2}m+h-d}$, which implies $n\geq 2^{\frac{7m+2h-2d}{d+2}}$. When parameters satisfy this bound, Quantum Algorithm 6 holds the potential to outperform Classical Algorithm 1. Subject to the actual term parameter $t$, the dominant term of Quantum Algorithm 6 reduces to $T_{q_2}'=2^{\frac{7}{2}m+h}n^{\frac{5}{2}t}$. Considering $\frac{T_{c_1}}{T_{q_2}'}\geq 1$, when parameters satisfy $n\geq 2^{\frac{7m+2h-2d}{6d-5t+2}}$, Quantum Algorithm 6 holds a potential advantage.
\par
Comparing Classical Algorithm 2 with Quantum Algorithm 6, consider that the condition $\frac{T_{c_2}}{T_{q_2}}\geq 1$ is equivalent to $n^{\frac{d}{2}}\geq 2^{\frac{5}{2}m+h-d}$, which implies $n\geq 2^{\frac{5m+2h-2d}{d}}$. When parameters satisfy this bound, Quantum Algorithm 6 holds the potential to outperform Classical Algorithm 2. Fig.~\ref{fig6} validates this theoretical result. Restricted by the parameter $t$, considering $\frac{T_{c_2}}{T_{q_2}'}\geq 1$, when parameters satisfy $n\geq 2^{\frac{5m+2h-2d}{6d-5t}}$, Quantum Algorithm 6 holds a potential advantage.

\par 
\section{Conclusion and Future Work}
\label{section8}
In this paper, we conduct a systematic study of classical and quantum approaches to the LPSN problem in the sample-restricted regime. On the classical side, by combining the binary Schwartz--Zippel lemma with inequality bounding techniques, we derive for the first time the exact sample and time complexity upper bounds of the bit-by-bit guessing algorithm and the Gaussian-elimination-based algebraic algorithm under a success probability of $1-\varepsilon$, thereby explicitly delineating the theoretical boundaries of classical algebraic attacks.

On the quantum side, the central contribution of this paper is a novel polynomial system reduction method, Red2. This method equivalently transforms a nonlinear Boolean system of $r$ equations and $n$ variables into a $3$-sparse complex system whose constant terms are strictly $-1$, ensuring that the resulting Macaulay matrix satisfies the prerequisites for efficient quantum state preparation within $\mathcal{O}(T_F-2r)$ time. Building on this, we bridge the theoretical gap left by Ding et al. regarding unquantified reduction impacts, tightening the condition number lower bound of a Macaulay matrix of total degree $D$ from $\Omega\left(\left(\binom{D+h}{h}-1\right)^{\frac{1}{2}}\right)$ to $\Omega\left((T_F-2r)^{-\frac{1}{2}}\left(\binom{D+h}{h}-1\right)^{\frac{1}{2}}\right)$, and further extend the framework to Boolean Macaulay systems, decreasing the bound to $\Omega\left(T_F^{-\frac{1}{2}}(2^h-1)^{\frac{1}{2}}\right)$. Numerical instances constructed via an iterative optimization algorithm breach the theoretical limits established by Ding et al., empirically validating these tighter bounds.

We further carry out a fine-grained logical-level evaluation of the quantum resources required by the proposed algorithms, deriving closed-form expressions for the circuit width, depth, and gate count. Since the circuit size scales as the square of the condition number, the compression of the condition number interval achieved by Red2 translates directly into a reduction in circuit depth and gate count, demonstrating that our improvement is not merely asymptotic but yields concrete, quantifiable resource savings. Applying the improved quantum algorithms to the LPSN problem exponentially reduces the sample complexity; combined with 3D time complexity surfaces and an algorithm selection strategy, we show that even when $h=\Theta(\sqrt{n})$, the quantum algorithm retains the potential to outperform classical algorithms within specific parameter regimes.

In summary, this research provides solid algorithmic and theoretical support for post-quantum algebraic cryptanalysis in sample-restricted scenarios. Future work may focus on quantum solving techniques that achieve optimal linear dependence on the condition number, such as the quantum singular value transformation, or explore heuristic quantum algorithms, such as the quantum approximate optimization algorithm, to solve nonlinear Boolean systems directly, thereby bypassing the linearization step.
\begin{acknowledgments}
The authors gratefully acknowledge the support of the National Cryptologic Science Fund of China (Grant No. 2025NCSF02010).
\end{acknowledgments}

\appendix

\section{Structured Approximations for\\ Low-Noise LPN Problems}
Many restrictions on the noise within the LPN problem can be delineated using the LPSN framework. In fact, assuming the number of correlated noise variables is $m$, any structure within this set of samples that excludes at least one noise vector $\bm{\eta}$ can be represented by a non-zero polynomial. Below, we present three commonly utilized structured noise patterns, which approximate the low-noise-rate LPN problem with high probability. We exclusively provide a detailed analysis for the first structure; the analytical methodologies for the remaining two structures are perfectly analogous.

\par
\begin{example}
Over $\mathrm{GF}(2)$, construct the polynomial $P(\bm{\eta})=1\oplus\prod_{1\leq i<j\leq m}(1\oplus\eta_i\eta_j)$. This polynomial over $\mathrm{GF}(2)$ satisfies the following property: when the Hamming weight of $\bm{\eta}$ (i.e., the number of components equal to 1) does not exceed 1, $P(\bm{\eta})=0$; otherwise, $P(\bm{\eta})=1$. The LPSN problem associated with the polynomial $P$ is equivalent to imposing an additional constraint on the noise in the LPN problem, requiring that the Hamming weight of the vector formed by the noise across $m$ samples is either 0 or 1. Intuitively, this corresponds to a special case of the LPN problem where the noise rate is strictly less than $\frac{1}{m}$. However, the noise rate designates the expected probability that a noise bit is 1; in practice, the Hamming weight of the noise across $m$ samples may exceed 1. By strictly forcing the noise in the $m$ samples to conform to an acceptable noise pattern—that is, satisfying $P(\bm{\eta})=0$—we eliminate the possibility of the noise's Hamming weight exceeding 1, thereby guaranteeing that the Hamming weight of the noise within these $m$ samples strictly does not exceed 1. Below, from a probabilistic perspective, we rigorously analyze the equivalence between the low-noise-rate LPN problem and the LPN problem under the aforementioned structured noise pattern.

\par
In the LPN problem with a noise rate of $p$, each sample independently contains noise with a probability of $p$. When we query the oracle $Q(n,p,\bm{s})$ exactly $m$ times, we obtain $m$ independent samples. The probability that at most one of these samples contains noise is given by $\Pr[h(\bm{\eta})\leq 1]=(1-p)^m+mp(1-p)^{m-1}$. We desire this probability to be sufficiently high to approximately simulate the LPSN oracle associated with the polynomial $P(\bm{\eta})=1\oplus\prod_{1\leq i<j\leq m}(1\oplus\eta_i\eta_j)$. Assuming $m$ is sufficiently large, when $p=\frac{1}{m}$, we have $\Pr[h(\bm{\eta})\leq 1]\approx 0.7358$. This indicates that the $m$ samples obtained via the LPN oracle $Q$ can serve as the output samples of the LPSN oracle associated with the polynomial $P$ with a probability of $0.7358$. 
In fact, the LPSN problem associated with the polynomial $P$ is equivalent to the LPN problem with a noise rate of $p\leq\frac{1}{10m}$ with a probability of $0.9953$. Naturally, the aforementioned high-probability equivalences consistently necessitate a substantially large number of correlated samples $m$. That is, when the correlation among sample noises is pronounced, algorithms designed for solving the LPSN problem can recover the secret vector in the LPN oracle corresponding to lower noise rates with a notably high success probability.
\end{example}
\par
\begin{example}Over $\mathrm{GF}(2)$, construct the polynomial $P(\bm{\eta})=1\oplus\sum_{i=1}^m\left(\eta_i\prod_{j\neq i}(1\oplus\eta_j)\right)$. This polynomial satisfies the property that $P(\bm{\eta})=0$ if and only if the Hamming weight of $\bm{\eta}$ is exactly 1. This implies that the LPSN problem associated with the polynomial $P$ demands that exactly one out of the $m$ samples contains noise.
\end{example}
\par
\begin{example}Over $\mathrm{GF}(2)$, construct the polynomial $P(\bm{\eta})=1\oplus\prod_{S\subseteq [m], |S|=w+1}\left(1\oplus\prod_{i\in S}\eta_i\right)$, where $w$ satisfies $2w<m$. This polynomial satisfies the property that $P(\bm{\eta})=0$ if and only if the Hamming weight of $\bm{\eta}$ does not exceed $w$. This indicates that the LPSN problem associated with the polynomial $P$ requires that at most $w$ out of the $m$ samples contain noise.
\end{example}
\par
It is straightforward to verify that the polynomials associated with the three structured noise patterns provided above all satisfy the requisite algebraic conditions.

\nocite{*}
 
\bibliography{mydata}

\end{document}